\documentclass[a4paper,twocolumn,11pt,accepted=2022-05-29]{quantumarticle}
\pdfoutput=1
\usepackage[utf8]{inputenc}
\usepackage[english]{babel}
\usepackage[T1]{fontenc}
\usepackage{color}
\usepackage{amsmath}
\usepackage{amssymb}
\usepackage{amsthm}
\usepackage{physics}
\usepackage{thmtools}
\usepackage{xcolor}

\usepackage[breaklinks]{hyperref}

\usepackage{tikz}
\usepackage{lipsum}

\usepackage[capitalise,compress]{cleveref}\crefname{section}{Sec.}{Secs.}
\Crefname{section}{Section}{Sections}
\crefrangelabelformat{equation}{\textup{(#3#1#4)}--\textup{(#5#2#6)}}
\newtheorem{thm}{Theorem}
\newtheorem{lem}{Lemma}

\crefname{cor}{Cor.}{Cors.} 
\DeclareMathOperator{\dist}{dist}

\DeclareMathOperator{\Cut}{\text{\rm Cut}}
\DeclareMathOperator{\e}{\text{\rm e}}

\begin{document}

\title{Limits of Short-Time Evolution of Local Hamiltonians}

\author{Ali Hamed Moosavian}
\email{ali@phanous.ir}
\affiliation{QuOne Lab, Phanous Research and Innovation Centre, Tehran, Iran}
\orcid{0000-0002-1818-5037}
\author{Seyed Sajad Kahani}
\email{sajad@phanous.ir}
\author{Salman Beigi}
\email{salman@phanous.ir}
\affiliation{QuOne Lab, Phanous Research and Innovation Centre, Tehran, Iran}
\orcid{0000-0003-3588-4662}

\maketitle

\begin{abstract}

Evolutions of local Hamiltonians in short times are expected to remain local and thus limited. 
In this paper, we validate this intuition by proving some limitations on  short-time evolutions of local time-dependent Hamiltonians. We show that the distribution of the measurement output of short-time (at most logarithmic) evolutions of local Hamiltonians are \emph{concentrated} and satisfy an \emph{isoperimetric inequality}. To showcase explicit applications of our results, we study the \normalsize{M}\scriptsize{AX}\normalsize{C}\scriptsize{UT} \normalsize{problem and conclude that  quantum annealing needs at least a run-time that scales logarithmically in the problem size to beat classical algorithms on }\normalsize{M}\scriptsize{AX}\normalsize{C}\scriptsize{UT}\normalsize{. To establish our results, we also prove a Lieb-Robinson bound that works for time-dependent Hamiltonians which might be of independent interest. }
\end{abstract}

\section{\label{sec:intro}Introduction}

The Quantum Annealing (QA) algorithm \cite{Kadowaki1998} was first introduced as a general purpose optimization algorithm for quantum computers in \cite{Farhi2000}. The QA algorithm is based on the adiabatic theorem in quantum mechanics \cite{Kato1950}. The adiabatic theorem states that if we initialize a system in its ground state, and then evolve the Hamiltonian of the system slowly enough, as long as the spectral gap of the Hamiltonian does not close, the system remains close to the ground state of the system at all times \cite{Born1928,Albash2018}. Because of its simple conceptual structure, QA has been implemented on many different platforms and applied to a variety of different problems, see \emph{e.g.}\@ \cite{hen2016, puri2017, lechner2015, jiang2018, li2018, stella2005, titiloye2011, mott2017, pudenz2015, perdomo2012, pudenz2014, martovnak2004, adachi2015, johnson2011, boixo2013}.

In this manuscript, we  focus on short-time QA algorithms. This setup is  motivated by the current state of the art quantum technologies (see methods of \cite{King2022}). On one hand, we still lack universal quantum computers capable of  simulating general QA algorithms \cite{Google,wright2019benchmarking}, and on the other hand the analog QA devices lack fault-tolerance. Therefore, the total time of reliably performing a continuous computation on these analog devices is bounded by their decoherence time. 

By considering short-time quantum annealing, we exit the scope of the adiabatic theorem. Hence, our results are not sensitive to the Hamiltonian gap. In this context, the algorithms considered in our work are sometimes called "diabatic quantum annealing" in the literature (E.\@g.\@ see \cite{Crosson2021} and references therein). 

Quantum Approximate Optimization Algorithm (QAOA)  is another general purpose quantum algorithm that can yield approximate solutions to many optimization problems, particularly combinatorial ones \cite{Farhi2014}. Based on QAOA's local nature, recent developments have shown that  low-depth QAOAs are unable to produce better approximate solutions to some optimization problems than the best known classical algorithms \cite{Farhi2020a,Farhi2020b,Bravyi2019,Bravyi2021,Bravyi2020}.

In particular, Bravyi et al \cite{Bravyi2019} use the \emph{isoperimetric inequality} of \cite{Eldar2017} that puts restrictions on the output distributions of all low-depth circuits, to show that low-depth (at most logarithmic) QAOA for \textsc{MaxCut} does not outperform best classical algorithms on \emph{Ramanujan graphs}. Also, Farhi et al~\cite{Farhi2020a} show that low-depth QAOA for the problems of \textsc{MaxCut} and Maximum Independent Set, does not perform well on random regular bipartite graphs. 

Thinking of the correspondence between short-time QA algorithms and low-depth quantum circuits, particularly low-depth QAOAs, it is natural to ask whether the above bounds hold for QA algorithms or not.  In this paper, we show that most of the aforementioned results on limitations of QAOAs still hold for short-time QA algorithms as well. Throughout this article we use short-time to represent times that scale at most logarithmically with the system size. We first show that the isoperimetric inequality of~\cite{Eldar2017} remains to hold for output distributions of short-time QA computations. Then, using this inequality as a tool (besides others), we generalize the results of~\cite{Bravyi2019} and~\cite{Farhi2020a} to the QA case.   

QA and QAOA can be thought of two instances of optimal control problems. In QA we look for an optimal path in the space of Hamiltonians from a starting Hamiltonian to a target one, while in QAOA we look for optimal rotation angles of some local circuits.  
Applying \emph{Pontryagin's maximum principle}, some authors have argued that the optimal solution to the first control theory problem takes the so called \emph{bang-bang} cond problem. However, it was shortly noted that in some common problem instances the control Hamiltonian of the system becomes \emph{singular}, a condition that was presumed to be extremely rare, and then the optimal control protocol is no longer of the bang-bang form~\cite{Brady2020,Brady2021,Venuti2021}. This means that, in general, the optimal control problem for QA does not reduce to that of the QAOA, and needs  thoughtful considerations. In fact, the observations of~\cite{Brady2020,Brady2021,Venuti2021} show that we cannot directly use the results on the obstructions of QAOAs, to prove the aforementioned limitations of short-time QA algorithms.  

Another idea for relating results on low-depth quantum circuits (QAOAs) to QA algorithms is \emph{trotterization} (see~\cite{Childs2020} and references therein), which approximates the unitary evolution of a local Hamiltonian with a quantum circuit. Although one can imagine this approach working in some very specific instances, in general the trotterized circuit must have an asymptotic depth that grows polynomially with the number of qubits. This prohibits the application of the aforementioned QAOA results, \cite{Farhi2020a,Farhi2020b,Bravyi2019,Bravyi2021,Bravyi2020}, as they only apply to logarithmic depths.

We use \emph{Lieb-Robinson bounds} as a main tool for proving our results \cite{Nachtergaele2006,Nachtergaele2010,Bravyi2006,Chen2019}. To this end, we first prove a Lieb-Robinson bound that works for time-dependent Hamiltonians and is amenable to QA algorithms. To the best of our knowledge, this is the first time-dependent Lieb-Robinson bound. 

As another technical tool, we use the notion of \emph{$\gamma_2$-norm}, also called \emph{factorization norm}, from matrix analysis to bound the derivative of certain functions on the space of operators. This enables us to bound certain errors induced by applying the Lieb-Robinson approximation.

Here is a summary of our results and the structure of the paper. In \cref{sec: Tools} we state our version of the Lieb-Robinson bound that is a fundamental tool in our work. Then in \cref{sec: concentration of measure} we outline the proof of two different bounds on the output distributions of short-time evolutions of local time-dependent Hamiltonians. Our first bound is a generalization of the well-known \emph{Chebyshev inequality} which says that the Hamming weight of the measurement outcome of a short-time local Hamiltonian evolution is \emph{concentrated} around its mean. The next bound, that is independent of the first one, is a generalization of the isoperimetric inequality of \cite{Eldar2017} for short-time Hamiltonian evolutions. Then,  in \cref{sec: Koenig}, we take the approach of \cite{Bravyi2019}, and show that the $\mathbb Z_2$-symmetry of the \textsc{MaxCut} Hamiltonian limits the performance of QA algorithms on Ramanujan graphs. To prove this result we use our isoperimetric inequality. Next, in \cref{sec: QA of Farhi and Gamarnik and Gutmann} , we prove a similar limitation on random regular bipartite graphs  using the Lieb-Robinson bound. This result is an extension of the result of \cite{Farhi2020a} for QA algorithms.
Final remarks come in \cref{sec: conclusion}, and many of the detailed proofs are left for Supplemental materials. 

\section{Lieb-Robinson Bound}\label{sec: Tools}

One of our main tools in this manuscript is the Lieb-Robinson bound \cite{Lieb1972}. Roughly speaking, it states that the evolution of a local observable in the Heisenberg picture, remains almost local in short times. We use the following version of the Lieb-Robinson bound that notably works for \emph{time-dependent Hamiltonians} as well. 

To state our theorem we need a definition. Given a graph $G$ and a subset of its vertices $A$, we denote
the set of vertices on the $L$-\emph{boundary} of $A$  by 
$\partial_L (A)$. That is,
\begin{align}
\partial_L (A)=& \{x\in A:\, \exists y\in A^c, \dist_G(x, y)\leq L\}\\
&~\cup \{x\in A^c:\, \exists y\in A, \dist_G(x, y)\leq L\},
\end{align}
where $A^c$ is the subgraph consisting of vertices not in $A$ and $\dist_G(x,y)$ is the graph distance between vertices $x$ and $y$.

\begin{thm}[Lieb-Robinson Bound]
    \label{thm: lieb-robinson}
    Let $G$ be a graph that may contain loops and parallel edges with maximum degree $\Delta>1$. 
    Let $\vb H(t)$ be a time-dependent Hamiltonian defined on vertices of $G$ of the form
    \begin{equation} 
    \label{eq:hamiltonian-t}
    \vb H(t) = \sum_{e}  u_e(t)\vb h_{e}, \end{equation}
    where $\vb h_{e}$ is time-independent and acts non-trivially only on the two ends of edge $e$ in $G$. Let $\vb U(t)$ be the unitary evolution associated with $\vb H(t)$. Suppose that $\norm{u_e(t)\vb h_{e}} \le g$ for any $e$ and $0\leq t\leq T$. Throughout this paper $\norm{.}$ represents the spectral norm unless it is clearly defined otherwise. Let $\vb O_A$ be an observable acting on region $A$ of the graph. For $L>1$ let 
    \begin{equation} \vb{\widetilde{H}}_A(t) := \sum_{e \subseteq A\cup \partial_L (A)} u_e(t)\vb h_e, \end{equation}
    be the Hamiltonian consisting of terms in the region \(A\cup \partial_L (A)\), and let $\vb V_A(t)$ be the unitary evolution associated to $\vb{\widetilde{H}}_A(t)$.
    Then we have
    \begin{widetext}
    \begin{equation} \label{eq: Local Operator Difference}
        \norm{\vb U^\dagger(T) \vb O_A\vb U(T)-\vb V^\dagger_A(T) \vb O_A\vb V_A(T)} \le \sqrt{\frac{2}{\pi}} \,|A| \cdot \|\vb O_A\| \cdot  \e^{-L\big(\log L - \log T - \log (4 g(\Delta-1))\big) - \frac12 \log L}\, . 
    \end{equation}
  \end{widetext}

\end{thm}

Intuitively speaking, this theorem says that, as $\vb O_A$ is local, it remains \emph{almost} local under a short-time evolution of a local Hamiltonian. Indeed, since $\widetilde{\vb H}_A(t)$ is non-trivial only in the region $A\cup \partial_L (A)$, then $\vb V_A(T)$ also acts non-trivially only on this region. Therefore, by this theorem as it is schematically depicted in \cref{fig: Lieb-Robinson}, $\vb U^\dagger(T) \vb O_A\vb U(T)$ is approximated by $\vb V^\dagger_A(T) \vb O_A\vb V_A(T)$ that acts only on $A\cup \partial_L (A)$.

\begin{figure}[ht]
\centering
\includegraphics[width=6cm]{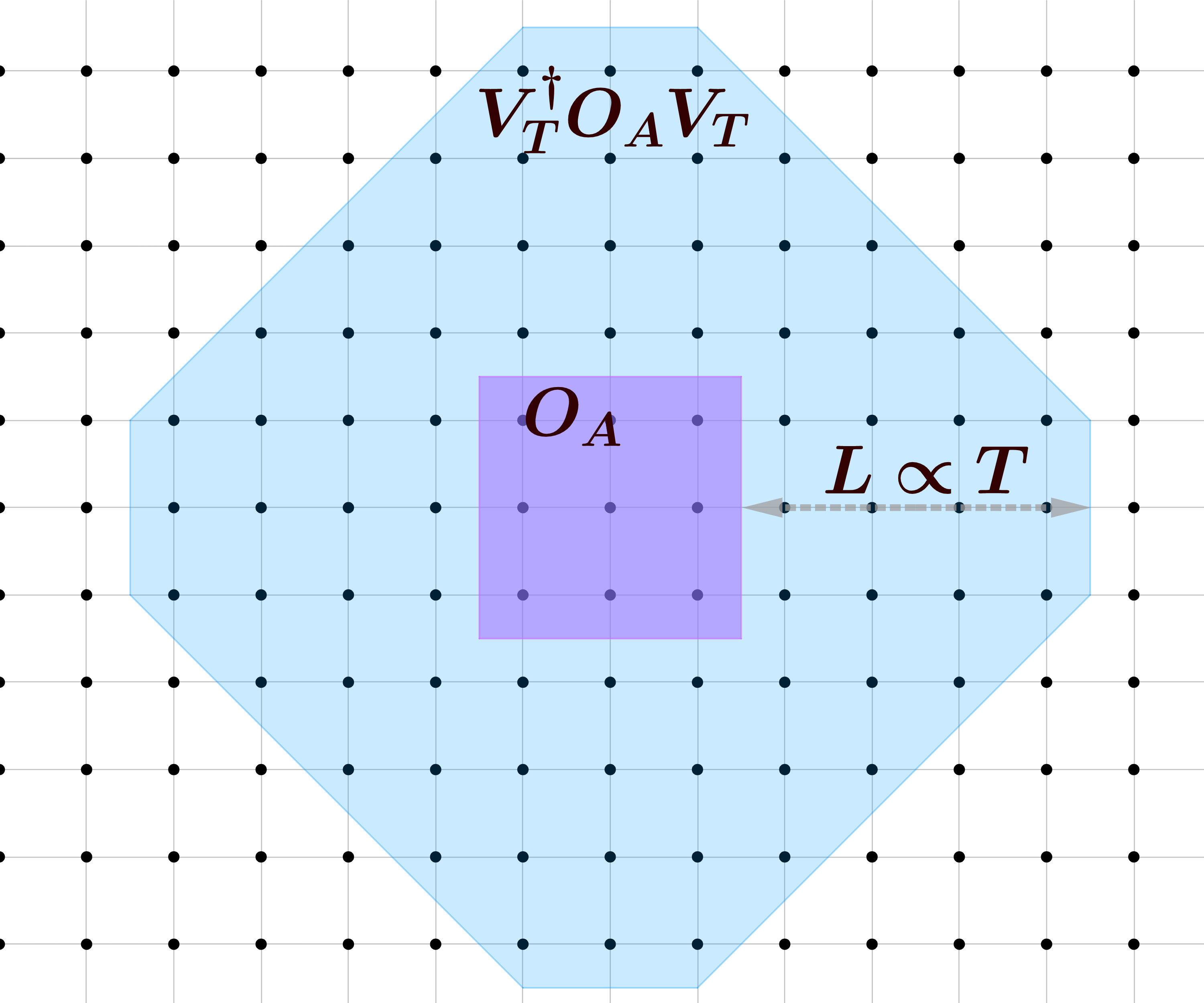}
\caption{(Color online) Illustration of  \cref{thm: lieb-robinson}. Here, $A$ is the region on which the operator $\vb O_A$ acts, and the larger shaded area represents $\partial_L(A)$. $\vb V^\dagger_A(T)\vb O_A \vb V_A(T)$ acts non-trivially on $A\cup \partial_L(A)$ and approximates $\vb U^\dagger (T) \vb O_A \vb U(T)$, the time evolution of $\vb O_A$ in the Heisenberg picture.}
\label{fig: Lieb-Robinson}
\end{figure}

Some remarks are in line regarding this theorem. First,  as the proof of this theorem given in \cref{app: Lieb-robinson} shows, the above Lieb-Robinson bound can easily be generalized for local Hamiltonians with $k$-body interactions, by replacing the term $\log(4g(\Delta-1))$ on the right hand side of \cref{eq: Local Operator Difference} with $\log(2kg(\Delta-1))$. Second, this theorem in the case where $\vb H$ is time-independent is proven in~\cite{Haah2018,Nachtergaele2010,Bravyi2006}. Third, the type of time-dependent Hamiltonians (\Cref{eq:hamiltonian-t}) considered by this theorem are motivated by some quantum annealing applications in mind. Nevertheless, note that in principle any time-dependent Hamiltonian can be written in this form; simply consider the expansion of each local term of the Hamiltonian in the Pauli basis.

\section{Limits on the output distribution of  short-time evaluations}\label{sec: concentration of measure}

In this section we present our main results, namely, how the output distribution of a short-time quantum annealing algorithm is constrained. Our first main result is a generalization of the well-known \emph{Chebyshev inequality} in probability theory that is a concentration of measure inequality. Our second result is an \emph{isoperimetric inequality}. A similar result for low-depth circuits has been proven in~\cite{Eldar2017}. 

In the following we assume that $\vb H(t)$ is a time-dependent Hamiltonian defined on the vertices of a graph $G$ as in \cref{eq:hamiltonian-t}. We assume that starting with a \emph{product state} 
\begin{align}\label{eq:psi-0-product}
\ket{\psi_0}=\ket{s_1}\otimes \ket{s_2}\otimes\cdots \otimes\ket{s_n},
\end{align}
the system evolves under the Hamiltonian $\vb H(t)$. Let
\begin{align}
    \ket{\psi_t} = \vb U(t) \ket{\psi_0},
\end{align}
where as before $\vb U(t)$ is the associated unitary evolution. We assume that at time $T$ we make a measurement in the computational basis to a get a distribution $p(\vb x)$ on $\{0,1\}^n$. 

\subsection{A Chebyshev bound\label{subsec: concentration}}

\begin{thm}[Quantum Chebyshev bound]\label{thm:Chebyshev}
    Let $\vb H(t)$ be a Hamiltonian on a graph with maximum degree $\Delta>1$ as in \cref{eq:hamiltonian-t} with $\|u_e(t)\vb h_e\|\leq g$ for any edge $e$. Let 
    \begin{align}
    T\le \frac {\kappa_1}{8g\Delta^{\frac{2-\kappa_1}{\kappa_1}}\log \Delta}\log n,
    \end{align}
    where $0<\kappa_1<1$ is an arbitrary constant. Let $p(\vb x)$ be the output distribution of measuring $\ket{\psi_T} = \vb U(T)\ket{\psi_0}$ in the computational basis. We assume $\ket{\psi_0}$ to be a product state as in \cref{eq:psi-0-product}. Then for any $c>0$ and $\kappa_1/2<\kappa_2<1/2$ we have
    \begin{equation}\begin{aligned}
        & \Pr\left[\left|w_H(\vb x)-\left(\frac n 2 - \frac m 2\right)\right|\ge cn^{\frac 1 2 + \kappa_2}\right] \\ &\le \frac{3}{2c^2}n^{-(2\kappa_2-\kappa_1)},
    \end{aligned}\end{equation}
    where $w_H(\vb x)$ is the hamming weight of $\vb x\in \{0,1\}^n$, i.e. number of non-zero coordinates of $\vb x$, and $m=\sum_i \bra{\psi_T}\mathbf Z_i \ket{\psi_T}$ with $\vb Z$ being the Pauli-$z$ operator.

\end{thm}

This theorem says that, when $T$ is bounded, with high probability the hamming weight of the measurement output is close to $\frac{1}{2}(n-m)$ which is the average of $w_H(\vb x)$. That is, the Hamming weight of the measurement outcome is \emph{concentrated}.

The proof idea of this theorem is similar to that of the standard Chebyshev inequality; we need to bound the variance of $w_H(\vb x)$. In the standard Chebyshev inequality the coordinates of $\vb x$ are independent while here they are correlated. Nevertheless, since $T$ is bounded, by the Lieb-Robinson bound (\cref{thm: lieb-robinson}) pairs of coordinates whose associated Hamming weights are far from each other are almost independent. This allows us to bound the variance of $w_H(\vb x)$. A detailed proof of this theorem is given in \cref{app: Chebyshev-bound}.

As a result of this theorem, as long as the initial state is a product state, no short-time annealing computation can produce the $n$-qubit GHZ state $\frac{1}{\sqrt 2}(\ket{0\cdots 0}+\ket{1\cdots 1})$ in sub-logarithmic times. This is because measuring this state in the computational basis gives two far apart peaks in the hamming weight $w_H$, \cref{fig: distribution}. We note that this statement holds for any annealing computation even with gapless Hamiltonians. Intuitively speaking, this is a result of the local nature of the Hamiltonian in \cref{eq:hamiltonian-t}.

\begin{figure}[ht]
\centering
\includegraphics[width=8.16cm]{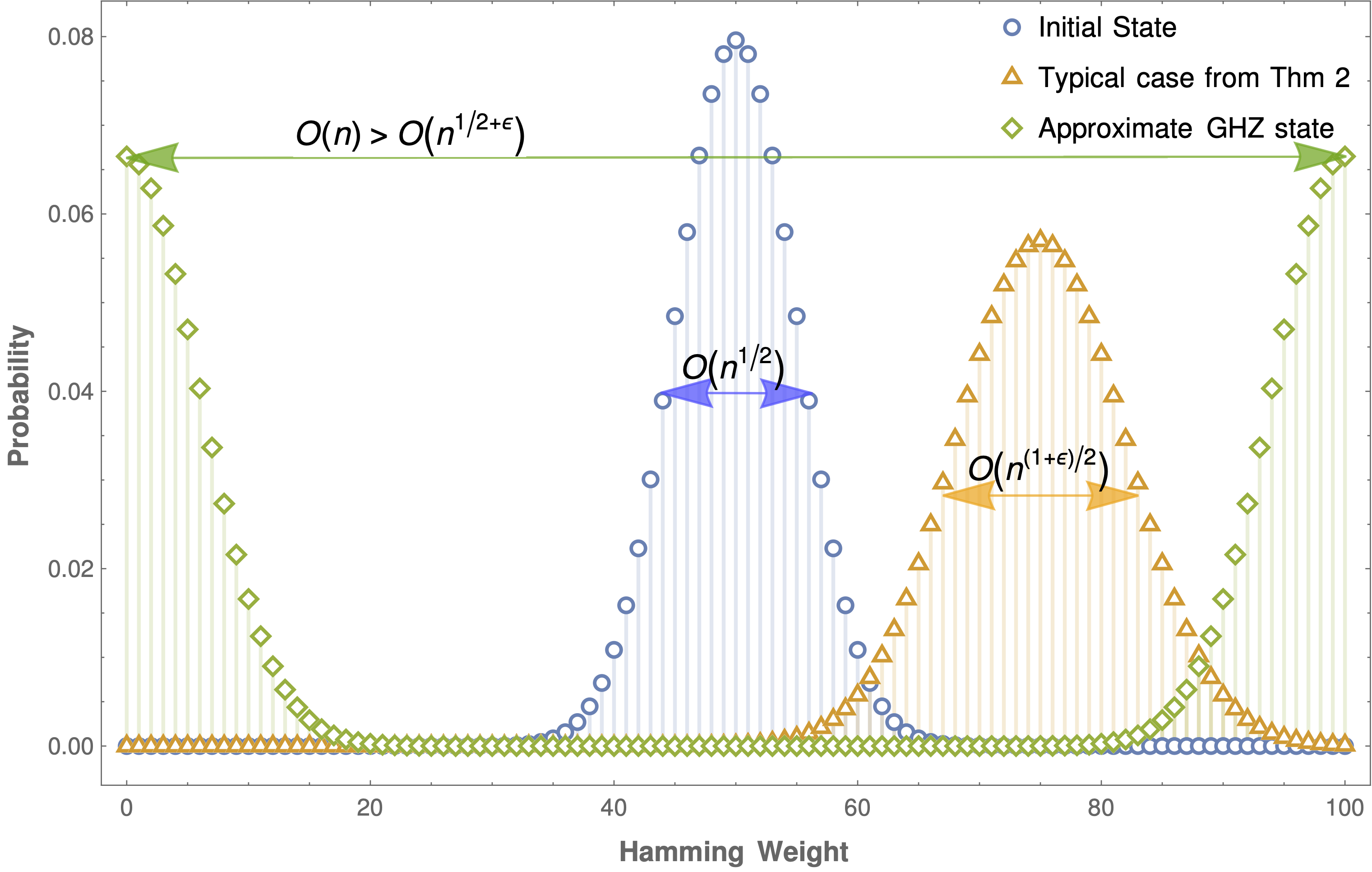}
\caption{(Color online) Illustration of Theorem~\ref{thm:Chebyshev} depicting the probability distribution of the Hamming weight of three different $n$-qubit states for $n=100$. The centered distribution in the middle (blue) is the typical distribution for a product initial state (e.g., $\ket{+}^{\otimes n}$). The distribution that is peaked just to the right (orange) represents the distribution after a short time evolution of the initial state described by \cref{thm:Chebyshev}. The distribution that has two peaks on either side of the spectrum (green) is an approximation for the GHZ state. This shows that the GHZ state cannot be generated even approximately by a short-time quantum annealing.}
\label{fig: distribution}
\end{figure}

\subsection{An isoperimetric inequality\label{subsec: Isoperimetric}}

We now prove that the output distribution of a short-time quantum annealing algorithm satisfies an \emph{isoperimetric inequality}, see \cref{fig: isoperimetric,fig:corollary}. Such a result for low-depth circuits was first proved in~\cite{Eldar2017}. 

\begin{thm}
    \label{thm: expansion-like}
    Let $\vb H(t)$ be a Hamiltonian on a graph with maximum degree $\Delta>1$ as in \cref{eq:hamiltonian-t} with $\|u_e(t)\vb h_e\|\leq g$ for any edge $e$. Let $\ket{\psi_0} = \ket{s_1}\otimes\ket{s_2}\otimes \cdots \otimes \ket{s_n}$ be an arbitrary product state and let $\ket{\psi_T}=\vb U(T)\ket{\psi_0}$ be the state of the system at time $T>0$.  Let
    $p(\vb x)$ be the distribution on $\{0,1\}^{n_0}$ resulting from measuring the first $ n_0 $ qubits of $\ket{\psi_T}$ in the computational basis, \emph{i.e.},
    \begin{align}
        p(\vb x) &= \sum_{\vb y\in\{0,1\}^{n- n_0 }} \abs{\braket{\vb x,\vb y}{\psi_T}}^2.
\end{align}
Then for any $F \subset \{0, 1\}^ {n_0} $ with $p(F) \le \frac{1}{2}$ and any $0\leq \theta\leq 1/2$ the following results hold:
\begin{itemize}

\item[{\rm (i)}]
Assume that 
\begin{equation}
T\leq \frac{\kappa_1}{4g\log (\Delta)\Delta^{\frac{1+2\kappa_1+\kappa_2}{\kappa_1}}}\log n_0\, , \label{eqn: optimal time limit}
\end{equation}
for some constants $\kappa_1, \kappa_2>0$. 
Then we have
\begin{align}    
    p\big(\partial_\ell (F)\big) \geq  \frac{1}{8}(2n_0^{1+\kappa_1})^{-2\theta} p(F)  - 2n_0^{-\kappa_2}.
\end{align}
where 
\begin{equation}
\ell = \frac{\kappa_1}{\sqrt 2\log (\Delta)} \log(n_0) n_0^{\big(\frac 12 -\theta\big)(1+\kappa_1)} \, . \label{eqn: ell}
\end{equation}

\item[{\rm (ii)}] Assume that we measure all the qubits, \emph{i.e.}\@ $n_0=n$. Moreover, assume that 
\begin{equation}\label{eq:bound-T-n_0=n}
T\leq \frac{\kappa_1}{4g(\Delta-1)}\log n\, , \end{equation}
and 
\begin{align}\label{eq:ell:n_0=n}
    \ell = \frac{\kappa_1(1+\kappa_2)}{2}\log(n)\sqrt n
\end{align}
for some constants $\kappa_1, \kappa_2>0$. Then we have
\begin{equation}p(\partial_\ell(F))\geq \frac{1}{8} p(F) - \frac 54 n^{-\big(\kappa_1(1+\kappa_2)\log(1+\kappa_2)-1\big)}.\end{equation}

\end{itemize}
\end{thm}

This theorem says that, measuring in terms of the outcome probabilities of a short-time quantum annealing computation, the size of a boundary of a set is large comparing to the size of a set itself, see Fig.~\ref{fig: isoperimetric}. The proof of this theorem is given in \cref{app: proof-expansion-like}.

\begin{figure}[t]
\centering
\includegraphics[width=3.82cm]{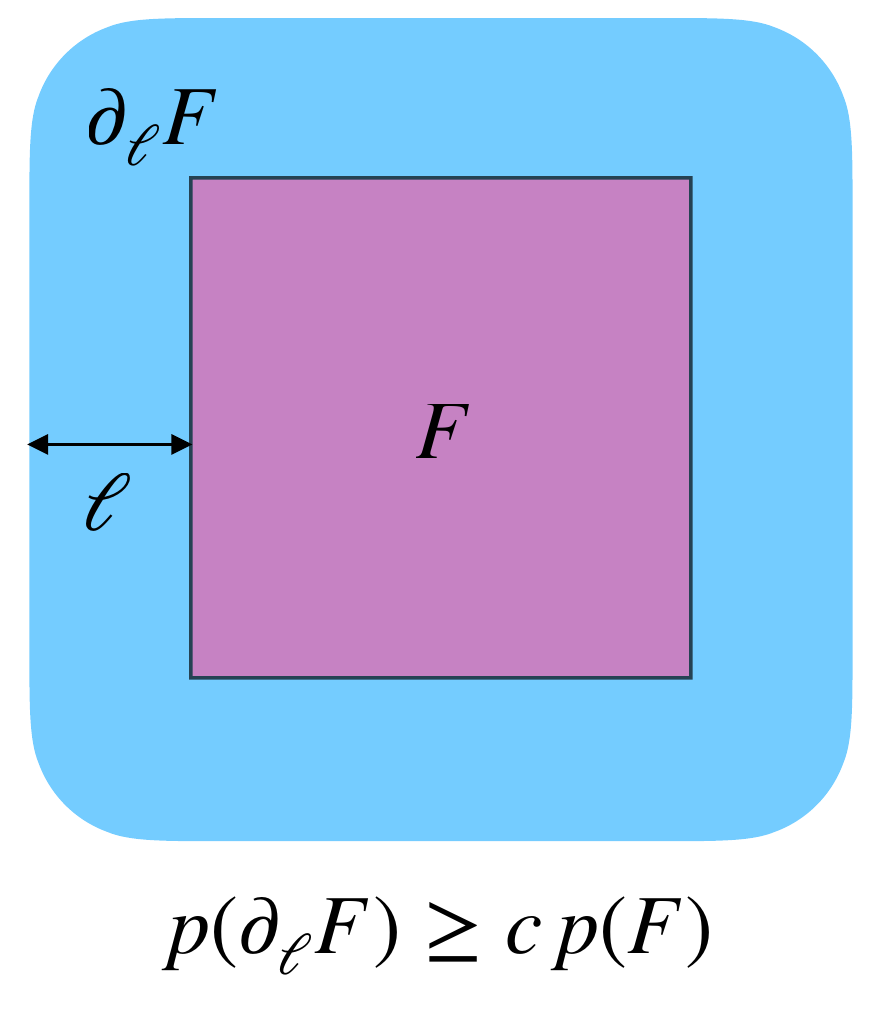}
\caption{(Color online) A schematic presentation of the isoperimetric inequality, stating that the boundary of a region is large comparing to the region itself.}
\label{fig: isoperimetric}
\end{figure}

\begin{figure}[ht]
\centering
\includegraphics[width=8.16cm]{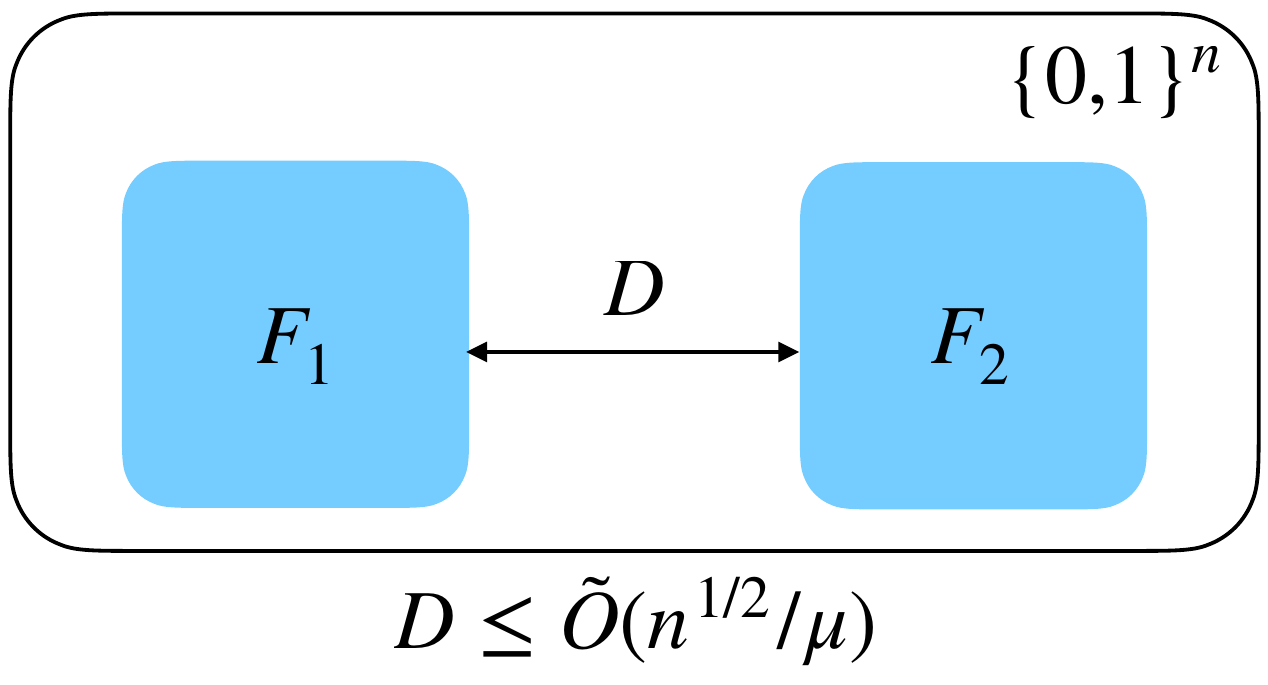}
\caption{(Color online) \Cref{thm: layers} says that measuring in terms of the output distribution of a short-time QA algorithm, if two sets have size at least $\mu$, then their distance (ignoring some logarithmic factors) is at most $n^{1/2}/\mu$.}
\label{fig:corollary}
\end{figure}

\begin{thm}\label{thm: layers}
    Let $\vb H(t)$ be a Hamiltonian on a graph with maximum degree $\Delta>1$ as in \cref{eq:hamiltonian-t} with $\|u_e(t)\vb h_e\|\leq g$ for any edge $e$. Let $\ket{\psi_0} = \ket{s_1}\otimes\ket{s_2}\otimes \cdots \otimes \ket{s_n}$ be an arbitrary product state and let $\ket{\psi_T}=\vb U(T)\ket{\psi_0}$ be the state of the system at time $T>0$.  Let
    $p(\vb x)$ be the distribution on $\{0,1\}^{n}$ resulting from measuring qubits of $\ket{\psi_T}$ in the computational basis.
Moreover, assume that 
\begin{equation}\begin{aligned}\label{eq:bound-T-n_0=n-b}
T & \leq \frac{\kappa_1}{4g(\Delta-1)}\log n,\, \\ \ell & = \frac{\kappa_1(1+\kappa_2)}{2}\log(n)\sqrt n,
\end{aligned}\end{equation}
for some constants $\kappa_1, \kappa_2>0$. Then for two arbitrary disjoint sets $F_1, F_2 \subset \{0, 1\}^ {n} $ with Hamming distance $D = \dist_H(F_1, F_2) $ and $\mu := \min \{p(F_1), p(F_2) \}$, we have
\begin{align}
    D&< \frac{16\ell}{\mu} + \frac{10D}{\mu} n^{-\big(\kappa_1(1+\kappa_2)\log(1+\kappa_2)-1\big)}
    \\
    &\leq \frac{16\ell}{\mu} + \frac{10}{\mu} n^{-\big(\kappa_1(1+\kappa_2)\log(1+\kappa_2)-2\big)}.
\end{align}

\end{thm}
The contra-positive version of this theorem says that if two far apart subsets of $\{0,1\}^n$ have high probability under the distribution produced by the annealing algorithm, then the run-time of the annealing algorithm is necessarily long.

Here is a high level perspectice of how the proofs are laid out in the appendix. In order to prove \cref{thm: expansion-like}, we need to bound the derivative of a certain Chebyshev polynomial of some quantum operators. Because of the approximate nature of the Lieb-Robinson bound, proving this requires an accurate mathematical analysis of the polynomials using the $\gamma_2$-norm. More naive approaches would have disproportionately large errors. The proof is laid out in \cref{app: gamma2-for-polys}. Then, by utilizing these results we prove \cref{thm: expansion-like} in \cref{app: proof-expansion-like}. \Cref{thm: layers} is basically a corollary of \cref{thm: expansion-like} and it is proven in \cref{app:layers}.

\section{\Large{M}\normalsize{AX}\Large{C}\normalsize{UT}\Large{ on Ramanujan Graphs}} \label{sec: Koenig}

In this section and the following one, we present limits of short-time quantum annealing algorithms for the \textsc{MaxCut} problem. Here, inspired by~\cite{Bravyi2019} 
our strategy is to
use the $\mathbb Z_2$-symmetry of the \textsc{MaxCut} Hamiltonian, and show that short-time quantum annealing algorithm cannot outperform best classical algorithms for \textsc{MaxCut} on Ramanujan graphs~\cite{Lubotzky1988}.

We first describe the quantum annealing algorithm for \textsc{MaxCut}.

Given a graph \(G = (V, E)\), the \textsc{MaxCut} problem asks for a \emph{cut} with maximum size in $G$; meaning what is the maximum number of edges between vertices in $V_{+}$ and  $V_{-}$ where $V_{+}\cup V_{-}=V$ form a partition of vertices? Labeling vertices in $V_+$ with $+1$ and vertices in $V_-$ with $-1$, this problem is equivalent to minimizing the energy of the following Hamiltonian:
     \begin{equation} \vb C := -\frac 1 2 \sum_{(i, j) \in E} {\left(\vb I-\vb Z_i \vb Z_j\right)}, \end{equation}
Where $\vb Z_i$'s are the Pauli-$z$ operators. Indeed, letting $\Cut^*$ be the maximum cut of $G$ we have
\begin{equation}\Cut^* = -\min_{\ket\phi} \bra \phi \vb C\ket \phi.\end{equation}

We may use quantum annealing to find the ground state and minimum energy of $\vb C$. To this end, a common choice for the \emph{driving Hamiltonian} is the summation of Pauli-$x$ operators:
    \begin{equation}
        \vb B := -\sum_{k \in V} \vb X_k.
    \end{equation} 
Using a control parameter $u : [0, T] \to [0, 1]$ with $u(0)=0$ and $u(T)=1$, the resulting time-dependent Hamiltonian is given by
    \begin{equation}
    \vb H(t) := \left(1-u(t)\right) \vb B + u(t) \vb C \label{Eq: MAXCUT Hamiltonian} \end{equation}
At time $t=0$, we start with the ground state of $\vb H(0)= \vb B$ that is equal to  
\begin{align}\label{eq:maxcut-psi-0}
\ket{\psi_0} = \ket{+}^{\otimes n},
\end{align}
where $n=|V|$ is the number of vertices. In the following we show that if $T$ is bounded, then 
\begin{equation}\bra{\psi_T} \vb C\ket{\psi_T},\end{equation}
is bounded away from the ground energy, i.e., $-\Cut^*$, for a certain family of graphs to be defined.

\medskip
For a graph $G=(V,E)$ on $n$ vertices, its \emph{Cheeger constant} is defined by~\cite{Mohar1989}:
\begin{equation}
    h(G) = \min_{S\subset V, 0<\abs{S}\le \frac n 2} \frac{\abs{\partial^E(S)}}{\abs{S}} \, ,
\end{equation}
where $\partial^E(S)$ is the set of \emph{edges} with one end in $S$ and one end outside of $S$. \emph{Ramanujan graphs} are certain families of graphs that have a large Cheeger constant; for a $\Delta$-regular Ramanujan graph $G$ we have
\begin{equation}h(G)\geq \frac{1}{2}\big(\Delta-2\sqrt{\Delta-1}\big).\end{equation}
It is known that~\cite{Marcus2015,Marcus2018,Hall2018} for any $\Delta\geq 3$ there exists an infinite family of $\Delta$-regular \emph{bipartite} Ramanujan graphs. We note that for such a graph, as a  bipartite graph, we have 
\begin{equation}\Cut^*= |E| = \frac{1}{2} \Delta n.\end{equation}

\begin{thm}\label{thm:MaxCut-Ramanujan}
Let $G$ be a $\Delta$-regular bipartite Ramanujan graph on $n$ vertices. Consider the quantum annealing algorithm given by \cref{Eq: MAXCUT Hamiltonian,eq:maxcut-psi-0}. Assume that
$T$ satisfies
\begin{equation}\label{eq:T-maxcut-bound}
T\leq \frac{\kappa_1}{4\Delta}\log n\, , \end{equation}
for some $\kappa_1>0$. 
Then, for any $0< \alpha< 1/2$, $0<\epsilon<1$ and sufficiently large $n$ we have
\begin{equation}\label{eq:cut-bound}
   \frac{-\bra{\psi_T}\vb C\ket{\psi_T}}{\Cut^*}< 1-\alpha(1-\epsilon) +2\alpha(1-\epsilon) \frac{\sqrt{\Delta - 1}}{\Delta}.
\end{equation} 
\end{thm}

We note that for appropriate choices of $\alpha, \epsilon$ and for sufficiently large (but constant) $\Delta\ge6$, the right hand side of \cref{eq:cut-bound} can be made smaller than $0.87856$. This means that the quantum annealing algorithm with $T\leq O(\log n)$  does not improve on the classical algorithm of Goemans and Williamson for \textsc{MaxCut}~\cite{goemans1995}. 

As the proof of this theorem in \cref{app:MaxCut-Ramanujan} shows, the smallest $n$ for which the bound \cref{eq:cut-bound} holds can be explicitly estimated in terms of $\alpha, \epsilon, \kappa_1$ and $\Delta$.

We remark that, comparing to \cref{eq:bound-T-n_0=n},  in \cref{eq:T-maxcut-bound} the parameter $\Delta$ is replaced with $\Delta+1$. The point is that in the \textsc{MaxCut} Hamiltonian \cref{Eq: MAXCUT Hamiltonian}, the number of terms acting on a vertex is $\Delta+1$ due to the extra term coming from $\vb B$.

\section{\large{M}\normalsize{AX}\large{C}\normalsize{UT}\large{ on Random Bipartite Graphs}}\label{sec: QA of Farhi and Gamarnik and Gutmann}

Inspired by the recent work of Farhi, Gamarnik and Gutmann~\cite{Farhi2020a}, in this section we claim that short-time QA algorithm for \textsc{MaxCut} does not perform well on random bipartite graphs. The idea is that such a short-time QA algorithm \emph{does not see the whole graph}, so it cannot distinguish between a bipartite random graph and an arbitrary random graph. Thus, the output cut-size of these algorithms on random bipartite graphs is bounded by the \textsc{MaxCut} of arbitrary random graphs.

\begin{thm}\label{thm:random-bipartite}
Consider the QA algorithm using the Hamiltonian~\eqref{Eq: MAXCUT Hamiltonian} with the initial state $\ket{\psi_0} = \ket{+}^{\otimes n}$, and time 
\begin{equation}T\leq \frac{\kappa}{16\log(\Delta)\Delta^{1+4/\kappa}}\log n\, ,\end{equation}
for some $0<\kappa<1$. Then, for any $\nu>0$ and sufficiently large $n$, the expected cut size from applying the algorithm on a random $\Delta$-regular bipartite graphs is at most
\begin{align}\label{eq:random-bipartite-bound-0}\left(\frac{\Delta}{4}+O(\sqrt{\Delta})+2\nu \right)n\, .
\end{align}
\end{thm}

Note that, by the bound in \cref{eq:random-bipartite-bound-0}, the performance of the QA algorithm on random $\Delta$-regular bipartite graphs is not much better than outputting a random cut.

The proof of this theorem is given in \cref{app:random-bipartite}.

Farhi et al also show how the low-depth QAOA algorithm cannot outperform classical algorithms for the maximum independent set problem on random bipartite graphs \cite{Farhi2020a}. Using similar arguments to what we provided in this section, that result can also be generalized to short-time QA algorithms. 

\section{Conclusions}\label{sec: conclusion}

Motivated by short-time quantum annealing, in this manuscript we proved limits on the output distribution of local time-dependent Hamiltonian evolutions. 
In particular, we showed that the distribution of the measurement outcome of a local Hamiltonian evolution is concentrated and satisfies an isoperimetric inequality if the evolution time is short. 
Using our framework, we  generalized results about limitations of low-depth quantum circuits, to short-time QA algorithms. To this end, an essential tool that we developed, is a Lieb-Robinson bound that works for time-dependent Hamiltonians and is amenable to QA algorithms. 

Our results are general in the sense that they work for the circuit model as well; they imply limitations on the output distribution of low-depth circuits with a locality structure. The point is that we can assume that the control parameters $u_e(t)$ in~\cref{eq:hamiltonian-t} have a so-called bang-bang form: they are either zero or take a fixed value, and for each $t$ the values of $u_{e_1}(t), u_{e_2}(t)$ for intersecting edges $e_1, e_2$ cannot both be non-zero. With such bang-bang control parameters, low-depth local circuits can be seen as special cases of short-time evolutions of local Hamiltonians for which our results hold.

The short-time condition is strictly essential for the bounds we have proven, and our results do not hold without this assumption. Indeed, there are several known examples where quantum annealing with $\textup{Poly}(n)$ run-time can provide superpolynomial speedups. For example, see \cite{Somma2012,Hastings2020,Gilyen2021}.

We believe that the techniques presented here can be applied to generalize other results about the limitation of low-depth circuits too, see \emph{e.g.}~\cite{Bravyi2021, Farhi2020b}, for short-time QA algorithm. We leave this for future research.

\bibliographystyle{quantum}

\onecolumn
\appendix

\section{Proof of the Lieb-Robinson Bound}\label{app: Lieb-robinson}

In this appendix we give a proof of \cref{thm: lieb-robinson}. As mentioned before, the proof ideas are essentially the same as those of the original Lieb-Robinson bound. Here, for simplicity of presentation of the proof we first give two lemmas.

\begin{lem}\label{lem:dif-eq-bound}
    (\cite{Nachtergaele2010}) Let $\vb X=\vb X(t)$ and $\vb Y=\vb Y(t)$ be time-dependent bounded operators with $\vb X$ being Hermitian. Let $\vb F=\vb F(t)$ be given by
    \begin{equation} \frac{\dd }{\dd t} \vb F = i[\vb X, \vb F] + \vb Y. \end{equation}
    Then we have
    \begin{equation} \frac{\dd}{\dd t} \|\vb F(t)\| \le \|\vb Y(t)\|. \end{equation}
\end{lem}
\begin{proof}
Writing the definition of the derivative we have
\begin{equation}
\begin{aligned}
\frac{\dd}{\dd t}\|\vb F\| &= \lim_{h\to 0} \frac{\|\vb F(t+h)\| - \|\vb F(t)\|}{h} \\
&= \lim_{h\to 0} \frac{\|\vb F(t) + ih[\vb X(t), \vb F(t)] + h\vb Y(t)\| + O(h^2) - \|\vb F(t)\|}{h} \\
&= \lim_{h\to 0} \frac{\|(I + ih\vb X)\vb F(I - ih\vb X) + ih^2\vb X\vb F\vb X + h\vb Y\| - \|\vb F\|}{h} \\
&= \lim_{h\to 0} \frac{\|\e^{ ih\vb X}\vb F \e^{- ih\vb X} + ih^2\vb X\vb F\vb X + h\vb Y\| + O(h^2) - \|\vb F\|}{h} \\
&\le \lim_{h\to 0} \frac{\|\e^{ ih\vb X}\vb F \e^{- ih\vb X}\| - \|\vb F\|}{h}  + h\|\vb X\vb F \vb X\|+ \|\vb Y\| \\
&= \|\vb Y\|\, ,
\end{aligned}
\end{equation}
where in the fourth line we use $\e^{\pm ih\vb X} = I\pm i h\vb X + O(h^2)$, and in the last line we use the fact that $\vb X$ is Hermitian and $e^{\pm ih \vb X}$ is unitary.
\end{proof}

Here is our second building block of the proof.

\begin{lem}
\label{lem:Liebunit}
    Using the notation of \cref{thm: lieb-robinson} and assuming that $L>1$ we have
    \begin{equation} \label{eq:lem-liebunit}
        \big\|\vb U(t)^\dagger \vb O \vb U(t) - \vb V(t)^{\dagger} \vb O\vb V(t)\big\| \le 2\|\vb O_A\| \sum_{e\sim A}\int_0^t |u_e(\tau)|\cdot  \big\|\vb U^\dagger(\tau) \vb h_e \vb U(\tau)- \vb V(\tau)^{\dagger}\vb h_e \vb V(\tau)\big\| \dd \tau\, ,
    \end{equation} 
    where \( \vb{\widetilde{H}}_A(t) := \sum_{e \subseteq A\cup \partial_L (A)} u_e(t)\vb h_e \) is the Hamiltonian consisting of terms in the region \(A \cup \partial_L (A)\), and $\vb V(t)$ is the unitary evolution associated to $\vb{\widetilde{H}}_A(t)$. Also, by $e\sim A$ we mean $e$ intersects $A$. Moreover, if $A=e_0$ for some edge $e_0$ and $\vb O_A=\vb  h_{e_0}$, then the sum in \cref{eq:lem-liebunit} can be further restricted to $e\neq e_0$.
\end{lem}
\begin{proof}
Let
\begin{equation}\vb F(t):=\vb U(t)^\dagger \vb O \vb U(t) - \vb V(t)^{\dagger} \vb O\vb V(t).\end{equation}
Then from Schrodinger's equation we have
\begin{equation} \label{eq:F(t)-derivative-6}\begin{aligned}
    \frac{\dd}{\dd t}\vb F(t) &= i \vb U(t)^{\dagger} \big[\vb H(t),\vb O\big ]\vb U(t) -  i \vb V(t)^{\dagger} \big[\vb{ \widetilde H}(t),\vb O\big ]\vb V(t)\\
        & = \sum_{e\sim A}  i \vb U(t)^{\dagger} \big[\vb h_e(t),\vb O\big ]\vb U(t) -  i \vb V(t)^{\dagger} \big[\vb h_e(t),\vb O\big ]\vb V(t)\\
    & = \sum_{e\sim A}  i  \big[\vb U(t)^{\dagger}\vb h_e(t)\vb U(t),\vb U(t)^{\dagger}\vb O\vb U(t)^{\dagger}\big ] -  i   \big[\vb V(t)^{\dagger} {\vb h}_e(t)\vb V(t),\vb V(t)^{\dagger}\vb O \vb V(t)\big ] \\
        & = \sum_{e\sim A}  i  \big[\vb U(t)^{\dagger}\vb h_e(t)\vb U(t),\vb F(t)+\vb V(t)^{\dagger}\vb O\vb V(t)^{\dagger}\big ] -  i   \big[\vb V(t)^{\dagger} {\vb h}_e(t)\vb V(t),\vb V(t)^{\dagger}\vb O \vb V(t)\big ] \\
            & =   i  \Big[\sum_{e\sim A} \vb U(t)^{\dagger}\vb h_e(t)\vb U(t),\vb F(t)\Big ] +  i \sum_{e\sim A}  \big[\vb U(t)^{\dagger} {\vb h}_e(t)\vb U(t)-\vb V(t)^{\dagger} {\vb h}_e(t)\vb V(t),\vb V(t)^{\dagger}\vb O \vb V(t)\big ] \, .\\
    \end{aligned} \end{equation}
Therefore, using $\vb F(0) =0$ and \cref{lem:dif-eq-bound} we can bound
\begin{equation} \begin{aligned}
\|\vb F(t)\| &= \left\|  \int_0^t \frac{\dd}{\dd \tau }\vb F(\tau) \dd \tau\right\| \\
&\leq \int_0^t \Big\|\frac{\dd}{\dd \tau }\vb F(\tau)\Big\| \dd \tau\\
&\leq \int_0^t \sum_{e\sim A}  \big\|\big[\vb U(\tau)^{\dagger} {\vb h}_e(\tau)\vb U(\tau)-\vb V(\tau)^{\dagger} {\vb h}_e(\tau)\vb V(\tau),\vb V(\tau)^{\dagger}\vb O \vb V(\tau)\big ]\big\| \dd \tau\\
&\leq 2\|\vb O\| \sum_{e\sim A}\int_0^t |u_e(\tau)|\cdot   \big\|\vb U(\tau)^{\dagger} {\vb h}_e\vb U(\tau)-\vb V(\tau)^{\dagger} {\vb h}_e\vb V(\tau)\big\|\dd \tau\, ,
\end{aligned} 
\end{equation}
where we use $\|[\vb X, \vb Y]\|\leq 2\|\vb X\| \|\vb Y\|$ in the last line.
Note that if $\vb O_A=\vb h_{e_0}$, in the second line of \cref{eq:F(t)-derivative-6} and thereafter, we can further restrict the sum to $e\neq e_0$.

\end{proof}

We can now give the proof of \cref{thm: lieb-robinson}. 
As before, let $\vb F(t):=\vb U(t)^\dagger \vb O \vb U(t) - \vb V(t)^{\dagger} \vb O\vb V(t)$. By \cref{lem:Liebunit} we have
\begin{equation}  
\|\vb F(T)\| \leq 2\|\vb O\| \sum_{e_1\sim A}\int_0^T |u_{e_1}(t_1)|\cdot   \big\|\vb U(t_1)^{\dagger} {\vb h}_{e_1}\vb U(t_1)-\vb V(t_1)^{\dagger} {\vb h}_{e_1}\vb V(t_1)\big\|\dd t_1\, .
\end{equation}
Assuming that $L>2$, we can apply \cref{lem:Liebunit} on each of the terms in the above sum, in which $A=e_1$ and $\vb O_A$ is replaced by $\vb h_{e_1}$. Therefore,
\begin{equation} \begin{aligned} 
\|\vb F(T)\| 
&\le 2\|\vb O_A\| \sum_{e_1 \sim A } \int_0^{T}  2\|u_{e_1}(t)\vb h_{e_1}\| \int_0^{t_1} \sum_{\stackrel{e_2 \sim e_1}{e_2\neq e_1}} |u_{e_2}(t_2)|\cdot \|\vb U(t_2)^\dagger \vb h_{e_2} \vb U(t_2) -\vb V(t_2)^\dagger \vb h_{e_2} \vb V(t_2)\| \dd t_2 \dd t_1 \\
&\le 2\|\vb O_A\|\, 2g \sum_{e_1 \sim A }\sum_{\stackrel{e_2 \sim e_1}{e_2\neq e_1}} \int_0^{T}  \int_0^{t_1}  |u_{e_2}(t_2)|\cdot \|\vb U(t_2)^\dagger \vb h_{e_2} \vb U(t_2) -\vb V(t_2)^\dagger \vb h_{e_2} \vb V(t_2)\| \dd t_2 \dd t_1 \, .
\end{aligned} \end{equation}
We can repeat the same process for $L$ times to obtain
\begin{equation} \begin{aligned} 
&\frac{\|\vb F(T)\|}{2\|\vb O_A\| (2g)^{L-1}}\\
&\le \sum_{\stackrel{(e_1, e_2, \dots, e_L):} {\text{ walk starting at }A} } \int_0^{T} \int_0^{t_1} \dots \int_0^{t_{L-1}} |u_{e_L}(t_L)|\cdot \|\vb U(t_L)^\dagger \vb h_{e_L} \vb U(t_L) -\vb V(t_L)^\dagger \vb h_{e_2} \vb V(t_L)\| \dd t_L \dots \dd t_2 \dd t_1\\
&\le \sum_{\stackrel{(e_1, e_2, \dots, e_L):} {\text{ walk starting at }A} } \int_0^{T} \int_0^{t_1} \dots \int_0^{t_{L-1}} 2g~ \dd t_L \dots \dd t_2 \dd t_1\\
& = 2g \frac{T^L}{L!} N_L\, ,
\end{aligned} 
\end{equation}
where $N_L$ denotes the \emph{number of walks} of length $L$ starting at $A$, i.e., $N_L$ is the number of sequences of edges $(e_1, \dots, e_L)$ such that $e_1\sim A$, $e_i\sim e_{i+1}$ and $e_i\neq e_{i+1}$ for $i=1, \dots ,L-1$. 
We note that as the maximum degree of the graph $G$ is $\Delta$, the number of choices for $e_1$ is at most $\Delta|A|\leq 2(\Delta-1)|A|$. Moreover, the number of choices for $e_{i+1}$ given $e_i$ is at most $2(\Delta-1)$.  Thus $N_L\leq (2(\Delta-1))^L|A|$. Therefore,
using the Stirling's approximation we have
\begin{equation} \begin{aligned}
\|\vb F(T)\| 
&\le 2|A| ~ \|\vb O_A\|  \frac{\big(4g(\Delta-1) T\big)^L}{L!} \\
&\le 2|A| \cdot \|\vb O_A\|  \big(4g(\Delta-1) T\big)^L \frac{1}{\sqrt{2\pi}}\e^{-L \log L + L - \frac{1}{2} \log L}\\
&\le \sqrt{\frac{2}{\pi}} |A| \cdot \|\vb O_A\|  \e^{-L\big(\log L - \log T - \log {(4g(\Delta-1))}\big) - \frac 12 \log L}\, .
\end{aligned} \end{equation}

\section{Proof of \texorpdfstring{\Cref{thm:Chebyshev}}{Theorem 3}}\label{app: Chebyshev-bound}
Let $\vb*\zeta=\sum_i \mathbf Z_i$ and use
Markov's inequality to obtain
\begin{align}
\text{Pr}\left[\left|w_H(\vb x)-\left(\frac n 2 - \frac m 2\right)\right|\ge cn^{\frac 1 2 + \kappa_2}\right] &= \text{Pr}\left[\left|\bra{\vb x}\vb*\zeta \ket{\vb x}- m\right|\ge 2cn^{\frac 1 2 + \kappa_2}\right] \\ 
&\le \frac{1}{\left(2cn^{\frac 1 2 + \kappa_2}\right)^2}\mathbb{E}\left[\big(\bra{\vb x}\vb*\zeta \ket{\vb x}- m\big)^2\right]\, .\label{eq:chebyshev-proof-markov}
\end{align}
Thus to prove the theorem it suffices to bound $\mathbb{E}\big[\big(\bra{\vb x}\vb*\zeta \ket{\vb x}- m\big)^2\big]$.

Let $m_i = \bra{\psi_T} \vb Z_i \ket{\psi_T}$. Then we have
\begin{equation}\begin{aligned}
    & \mathbb{E}\left[\left(\bra{\vb x}\vb*\zeta \ket{\vb x}- m\right)^2\right] \\
    &= \mathbb{E}\left[\left(\sum_i \bra{\vb x}\vb Z_i \ket{\vb x}- m_i\right)^2\right] \\ 
    &=\sum_{i,j}\mathbb{E}\big[\big(\bra{\vb x}\vb Z_i \ket{\vb x}- m_i\big)\cdot\big(\bra{\vb x}\vb Z_j \ket{\vb x}- m_j\big)\big]\\
    &=\sum_{i,j}\bra{\psi_T}(\vb{Z}_{i}-m_i)(\mathbf{Z}_{j}-m_j)\ket{\psi_T}\\
    &= \sum_{\substack{i,j \\ \dist_G(i,j)\leq 2L}}\bra{\psi_T}(\vb{Z}_{i}-m_i)(\mathbf{Z}_{j}-m_j)\ket{\psi_T}+ \sum_{\substack{i,j \\ \dist_G(i,j)> 2L}}\bra{\psi_T}(\vb{Z}_{i}-m_i)(\mathbf{Z}_{j}-m_j)\ket{\psi_T}\, ,
\end{aligned}
\end{equation}
where in the third line we use the fact that $\vb Z_i$ is diagonal in the computational basis, $\dist_G(i,j)$ denotes the distance of vertices $i, j$ on the graph $G$, and $L$ is given by 
\begin{align}\label{eq:chebyshev-L}
L= \frac {\kappa_1}{2\log \Delta}\log n\geq 4g \Delta^{\frac{2-\kappa_1}{\kappa_1}}T.
\end{align}
We analyze terms in the above sum separately. If $\dist_G(i, j)\leq 2L$, then by the 
Cauchy-Schwarz inequality we have
\begin{align}
    \bra{\psi_T}\left(\mathbf Z _i - m_i  \right)\left(\mathbf Z _j - m_j \right)\ket{\psi_T} &\leq \sqrt{\bra{\psi_T}\left(\mathbf Z _i - m_i \right)^2\ket{\psi_T}\cdot \bra{\psi_T}\left(\mathbf Z _j - m_j \right)^2\ket{\psi_T}}\\
    &=\sqrt{(1-m_i^2)(1-m_j^2)}\\
    &\leq 1\, .
\end{align}
Next, by the Lieb-Robinson bound (Theorem~\ref{thm: lieb-robinson}), for any $i$ there is an operator $\vb {\tilde Z}_i$ that acts only on $\partial_L(\{i\})$ and 
\begin{equation}\| \vb U(T)^\dagger (\vb Z_i-m_i) \vb U(T) - \vb {\tilde Z}_i\|\leq   \e^{-L\big(\log L - \log T - \log {4g\Delta}\big)}\leq \e^{-L\Big(2\frac{1-\kappa_1}{\kappa_1}\log\Delta\Big)}=  n^{-(1-\kappa_1)}\, ,\end{equation}
where the second inequality follows from~\eqref{eq:chebyshev-L}.
Therefore, letting $\vb \Upsilon_i=\vb U(T)^\dagger (\vb Z_i-m_i) \vb U(T) - \vb {\tilde Z}_i$, if $\dist_G(i, j)>2L$ we have
\begin{equation}\begin{aligned}
&\bra{\psi_T}\left(\mathbf Z _i - m_i  \right)\left(\mathbf Z _j - m_j \right)\ket{\psi_T} \\
& = \bra{\psi_0} \vb U(T)^\dagger\left(\mathbf Z _i - m_i  \right)\left(\mathbf Z _j - m_j \right) \vb U(T)\ket{\psi_0}\\
& = \bra{\psi_0}\left( \tilde{\mathbf Z}_i+\mathbf{\Upsilon}_i\right)\left( \tilde{\mathbf Z}_j+\mathbf{\Upsilon}_j\right)\ket{\psi_0} \\
&= \bra{\psi_0} \tilde{\mathbf Z}_i\tilde{\mathbf Z}_j\ket{\psi_0}+\bra{\psi_0}\mathbf{\Upsilon}_i \tilde{\mathbf Z}_j\ket{\psi_0}+\bra{\psi_0}\tilde{\mathbf Z}_i\mathbf{\Upsilon}_j\ket{\psi_0}+\bra{\psi_0}\mathbf{\Upsilon}_i\mathbf{\Upsilon}_j\ket{\psi_0} \\
&= \bra{\psi_0} \tilde{\mathbf Z}_i\ket{\psi_0}\bra{\psi_0}\tilde{\mathbf Z}_j\ket{\psi_0}+\bra{\psi_0}\mathbf{\Upsilon}_i \tilde{\mathbf Z}_j\ket{\psi_0}+\bra{\psi_0}\tilde{\mathbf Z}_i\mathbf{\Upsilon}_j\ket{\psi_0}+\bra{\psi_0}\mathbf{\Upsilon}_i\mathbf{\Upsilon}_j\ket{\psi_0} \, ,
\end{aligned}\end{equation}
where in the last line we use the facts that $\ket{\psi_0}$ is a product state and that the supports of $\tilde{\vb Z}_i, \tilde{\vb Z}_j$ do not intersect. Now we have
\begin{align}
    |\bra{\psi_0} \tilde{\mathbf Z}_i\ket{\psi_0}|& = |\bra{\psi_T} ({\mathbf Z}_i -m_i)\ket{\psi_T} -\bra{\psi_0} \vb\Upsilon_i\ket{\psi_0}|\\
    &= |\bra{\psi_0} \vb\Upsilon_i\ket{\psi_0}|\\
    &\leq  \|\vb \Upsilon_i\|\\
    & \leq  n^{-(1-\kappa_1)}.
\end{align}
Using this, and similar inequalities, we find that
\begin{align}
    \bra{\psi_T}\left(\mathbf Z _i - m_i  \right)\left(\mathbf Z _j - m_j \right)\ket{\psi_T}\leq 2 \left(  n^{-(1-\kappa_1)} \right)^2 + 2 n^{-(1-\kappa_1)}\leq 4 n^{-(1-\kappa_1)}.
\end{align}
Putting these together yields
\begin{align}
    \mathbb{E}\left[\left(\bra{\vb x}\vb*\zeta \ket{\vb x}- m\right)^2\right] &\leq \sum_{\substack{i,j \\ \dist_G(i,j)\leq 2L}} 1+ \sum_{\substack{i,j \\ \dist_G(i,j)> 2L}}4n^{-(1-\kappa_1)}\\
    & \leq 2n  \Delta^{2L} + n^2(4 n^{-(1-\kappa_1)})\\
    & = 6n^{1+\kappa_1},
\end{align}
where in the second line we use
\begin{align}\label{eq:neighborhood-L-bound}
1+\Delta\Big(1+(\Delta-1)+(\Delta-1)^2+\cdots (\Delta-1)^{2L-1}\Big)\leq 2\Delta^{2L}.
\end{align}
Using this in~\eqref{eq:chebyshev-proof-markov} the desired inequality follows.

\section{Derivative of a matrix function}\label{app: gamma2-for-polys}

In this appendix we compute a bound on the derivative of a certain polynomial function of matrices. This bound is an important ingredient in the proof of \cref{thm: expansion-like}. Indeed, in the proof, applying the Lieb-Robinson bound, a desired operator is approximated by a local one. Then, a polynomial function is applied on this operator, and the blow-up of the error of the Lieb-Robinson approximation induced by this polynomial should be estimated. We note that a naive approach to estimate this error is ineffective since the degree of this polynomial is large and some of its coefficients are exponentially large. Therefore, we need a more careful analysis of this error which itself demands a more careful analysis of the derivative of that polynomial as a map on the space of matrices.

 We first need to fix some notations. Given two matrices $\vb A$ and $\vb B$ of the same dimensions, we denote their \emph{entry-wise product} by $\vb A\circ \vb B$. In other words, $\vb A\circ \vb B$ is a matrix whose $(i, j)$-th entry equals the product of $(i, j)$-th entries of $\vb A$ and $\vb B$, i.e., $(\vb A\circ \vb B)_{ij} = \vb A_{ij}\cdot \vb B_{ij}$. The matrix $\vb A\circ \vb B$ is called the Schur or Hadamard product of $\vb A$ and $\vb B$.

For a continuously differentiable function $f:(a, b)\to \mathbb R$ and a sequence $\boldsymbol{\lambda}=(\lambda_1, \dots, \lambda_d)$ of numbers in $(a, b)$ we let $\mathcal D_{f, \boldsymbol{\lambda}}$ be the $d\times d$ matrix whose $(i,j)$-th entry equals
\begin{align}
    \big(\mathcal D_{f, \boldsymbol{\lambda}}\big)_{ij} = \begin{cases}
    \frac{f(\lambda_i)-f(\lambda_j)}{\lambda_i-\lambda_j} &\quad \lambda_i\neq \lambda_j,\\
    f'(\lambda_i)& \quad \lambda_i=\lambda_j,
    \end{cases}
\end{align}
where $f'(\lambda_i)$ is the derivative of $f$ at $\lambda_i$.
For a diagonal matrix $\vb \Lambda=\text{diag}(\boldsymbol{\lambda})$ with diagonal elements in $\boldsymbol{\lambda} =(\lambda_1,\dots, \lambda_d)$ we let 
\begin{align}
    \mathcal D_{f, \vb \Lambda} := \mathcal D_{f, \boldsymbol{\lambda}}.
\end{align}
Also, for a self-adjoint matrix $\vb A$, that can be diagonalized as 
$\vb A= \vb U \vb \Lambda\vb U^{\dagger}$ where $\vb U$ is unitary and $\vb \Lambda = \text{diag}(\boldsymbol{\lambda})$ is a diagonal matrix containing eigenvalues of $\vb A$, we let 
\begin{align}
\mathcal D_{f, \vb A} = \vb U \mathcal D_{f, \vb \Lambda} \vb U^{\dagger}.
\end{align}

\begin{lem}{\cite[pp.~124]{Bhatia1996}}\label{lem: Operator-difference}
Let $\vb A, \vb E$ be two Hermitian matrices and $f$ be a real continuously differentiable function on an interval containing the eigenvalues of $\vb A+t\vb E$ for all $t\in (-1, 1)$. Then we have   
\begin{equation}
    \dv{t} f\left(\vb A + t \vb E\right)\Big|_{t=0} = \mathcal D_{f,\vb A}\circ \vb E\, .
\end{equation}

\end{lem}

By this lemma, to bound the derivative of a matrix function as above, we need to bound the norm of the Schur product as a super-operator on the space of matrices. This super-operator norm is given by the so called $\gamma_2$-norm.

\begin{lem}{\cite[pp.~79]{bhatia2007}}
For a $d\times d$ matrix $\vb M$ let
\begin{align}\label{eq:def-gamma-2}
    \gamma_2(\vb M) := \inf\, \max\{\|\ket{v_i}\|^2:\, 1\leq i\leq d\}\cup \{\|\ket{w_j}\|^2:\, 1\leq j\leq d\},
\end{align}
where the infimum is taken over all vectors $\ket{v_i}, \ket{w_j}$ satisfying $\vb M_{ij} = \langle v_i\ket{w_j}$ for all $i, j$. Then for any $d\times d$ matrix $\vb E$ we have
\begin{align}
\|\vb M\circ \vb E\|\leq \gamma_2(\vb M) \|\vb E\|.
\end{align}
\end{lem}

It is not hard to verify that $\gamma_2$-norm as defined in \cref{eq:def-gamma-2} satisfies the triangle inequality and is a norm. We also note that if $\vb M_{ij}=\bra{v_i}w_j\rangle$ and $\vb M'_{ij}=\bra{v'_i}{w'_j}\rangle$ for two matrices $\vb M$ and $\vb M'$, then we have $(\vb M\circ \vb M')_{ij}= \bra{\tilde v_{i}}\tilde w_j\rangle$ where $\ket{\tilde v_i} = \ket{v_i}\otimes \ket{v'_i}$ and $\ket{\tilde w_j} = \ket{w_j}\otimes \ket{w'_j}$. Therefore, 
\begin{align}\label{eq:gamma-product}
    \gamma_2(\vb M\circ \vb M')\leq \gamma_2(\vb M)\cdot \gamma_2(\vb M').
\end{align}

The following lemma is a simple consequence of the lemmas above. 

\begin{lem}\label{lem: gamma-2-norm}
    For any two  Hermitian matrices $\vb A, \vb E$ and  polynomial $P$, we have
\begin{align}
    \norm{P\left(\vb A + \vb E\right)-P\left(\vb A\right)} &\le \|\vb E\|\cdot \int_0^{1}\gamma_2(\mathcal D_{P, \vb A+ t\vb E}) \dd t \,.
\end{align}
\end{lem}

\begin{proof}
By \cref{lem: Operator-difference} we have
\begin{align}
    \norm{P\left(\vb A +  \vb E\right)-P\left(\vb A\right)}  = \int_0^1 \dv{t} P(\vb A+t\vb E)\dd t
     = \int_0^1 \big(\mathcal D_{P,\vb A+t\vb E}\circ \vb E\big) \, \dd t
\end{align}
Then the result follows using \cref{lem: gamma-2-norm} and the triangle inequality.
\end{proof}

In the proof of \cref{thm: expansion-like} we will use this lemma for a polynomial that is defined in terms of the Chebyshev's polynomials. Recall that the Chebyshev's polynomial $T_n(x)$ is defined by 
\begin{align}\label{eq:Chebyshev-poly}
    T_n(\cos \theta) = \cos(n\theta).
\end{align}
This equation shows that $|T_n(x)|\leq 1=T_n(1)$ if $|x|\leq 1$. Moreover, $T_n(x)$ as a degree $n$ polynomial, has $n$ roots in the interval $[-1, 1]$. Thus, it is monotone in the interval $[1, \infty)$ and we have  $\left|T_n(x)\right| \le T_n(1+\delta)$ for $x \in [-1, 1+\delta]$. The Chebyshev polynomials can also be defined by recursive equations:
we have $T_0(x)=1$, $T_1(x)=x$ and 
\begin{align}
    T_{2n+1}(x) &= 2T_{n+1}(x)T_n(x)-x, \label{eq:chebyshev-odd}\\
    T_{2n}(x) &= 2T_n^2(x)-1 \, .\label{eq:chebyshev-even}
\end{align}

\begin{lem}\label{lem: gamma_2 inequality}
For any $0\leq \delta<1$ and any sequence $\boldsymbol{\lambda}=(\lambda_1, \dots, \lambda_d)$ with $-1 \leq \lambda_i\leq 1 + \delta$ we have
\begin{equation}
    \gamma_2(\mathcal D_{T_n,\boldsymbol{\lambda}}) \le (2n^2-1)T_n(1 + \delta), \qquad \forall  n\ge 1 \, .
\end{equation}
\end{lem}
\begin{proof}
    We prove the lemma by induction on $n$.
    For $n = 1$ it can be easily shown by setting $T_1(x) = x$.
    
    For an even number $2n$ we can write:
    \begin{align}
        \left(\mathcal D_{T_{2n},\boldsymbol{\lambda}}\right)_{ij} &=\frac{T_{2n}(\lambda_i)-T_{2n}(\lambda_j)}{\lambda_i-\lambda_j} \\
        &= 2\frac{T_{n}^2(\lambda_i)-T_{n}^2(\lambda_j)}{\lambda_i-\lambda_j} \\
        &= 2\frac{\left(T_{n}(\lambda_i)-T_{n}(\lambda_j)\right)\left(T_{n}(\lambda_i)+T_{n}(\lambda_j)\right)}{\lambda_i-\lambda_j} \\
        &= 2\left(\mathcal D_{T_n,\boldsymbol{\lambda}}\right)_{ij}\left(T_{n}(\lambda_i)+T_{n}(\lambda_j)\right) \\
        &= 2\left(\mathcal D_{T_n,\boldsymbol{\lambda}}\right)_{ij}\left(R_{T_{n},\boldsymbol{\lambda}} + R^{\intercal}_{T_{n},\boldsymbol{\lambda}}\right)_{ij} \, ,
    \end{align}
    where $\big(\mathcal R_{T_n, \boldsymbol{\lambda}}\big)_{ij} = T_n(\lambda_i)$, and $R^{\intercal}_{T_n,\boldsymbol{\lambda}}$ is its transpose. Now we can write:
    \begin{equation}
        \begin{aligned}
            \gamma_2(\mathcal D_{T_{2n},\boldsymbol{\lambda}}) &\le 2\gamma_2(\mathcal D_{T_{n},\boldsymbol{\lambda}}) \gamma_2(R_{T_{n},\boldsymbol{\lambda}} + R^{\intercal}_{T_{n},\boldsymbol{\lambda}})\, , \\
            &\le 2\left( (2n^2-1)T_n(1 + \delta) \right) \left( \gamma_2(R_{T_{n},\boldsymbol{\lambda}}) + \gamma_2(R^{\intercal}_{T_{n},\boldsymbol{\lambda}}) \right)\, .
        \end{aligned}
    \end{equation}
Using $\left|T_n(x)\right| \le T_n(1+\delta)$ for $x \in [-1, 1+\delta]$, we can find that
    \begin{equation}
        \gamma_2(R_{T_{n},\boldsymbol{\lambda}}), \gamma_2(R^{\intercal}_{T_{n},\boldsymbol{\lambda}}) \le T_n(1+\delta)\, .
    \end{equation}
 As a result
    \begin{equation}
        \gamma_2(\mathcal D_{T_{2n},\boldsymbol{\lambda}})  \le 4(2n^2-1) T_n^2(1+\delta)\, .
    \end{equation}
    
    Now recall that $T_n^2(1+\delta) = \frac{T_{2n}(1+\delta) + 1}{2}$ and knowing that $T_{2n}(1+\delta) \ge 1$
    \begin{equation}
        \begin{aligned}
            \gamma_2(\mathcal D_{T_{2n},\boldsymbol{\lambda}})  &\le 2(2n^2-1) \left( T_{2n}(1+\delta) + 1 \right)\, , \\
             &\le 4(2n^2-1) T_{2n}(1+\delta)\, , \\
             &\le \left(2(2n)^2-1\right) T_{2n}(1+\delta)\, . \\
        \end{aligned}
    \end{equation}
    
    Similarly for an odd number $2n+1$, we can write:
    \begin{align}
        \left(\mathcal D_{T_{2n+1},\boldsymbol{\lambda}}\right)_{ij} &= 2\frac{T_n(\lambda_i)T_{n+1}(\lambda_i)-T_n(\lambda_j)T_{n+1}(\lambda_j)}{\lambda_i-\lambda_j}-1 \\
        &= 2\frac{T_n(\lambda_i)T_{n+1}(\lambda_i)-T_n(\lambda_j)T_{n+1}(\lambda_i)+T_n(\lambda_j)T_{n+1}(\lambda_i)-T_n(\lambda_j)T_{n+1}(\lambda_j)}{\lambda_i-\lambda_j}-1 \\
        &= 2\left(\left(\mathcal{D}_{T_n,\boldsymbol{\lambda}}\right)_{ij}T_{n+1}\left(\lambda_i\right) + T_{n}\left(\lambda_j\right)\left(\mathcal{D}_{T_{n+1},\boldsymbol{\lambda}}\right)_{ij}\right)-1 \\ 
        &= 2\left(\left(\mathcal{D}_{T_n,\boldsymbol{\lambda}}\right)_{ij}\left((R_{T_{n+1},\boldsymbol{\lambda}}\right)_{ij} +\left(R^{\intercal}_{T_{n},\boldsymbol{\lambda}}\right) \left(\mathcal{D}_{T_{n+1},\boldsymbol{\lambda}}\right)_{ij}\right)-1 \, .
    \end{align}
    Now by bounding $\gamma_2(R_{T_{n},\boldsymbol{\lambda}})$ and $\gamma_2(R_{T_{n+1},\boldsymbol{\lambda}})$ 
    \begin{equation}
        \begin{aligned}
            \gamma_2(\mathcal D_{T_{2n+1},\boldsymbol{\lambda}}) \le&~ 2\gamma_2(\mathcal D_{T_{n},\boldsymbol{\lambda}}) \gamma_2\left((R_{T_{n+1},\boldsymbol{\lambda}}\right) +
            2\gamma_2(\mathcal D_{T_{n+1},\boldsymbol{\lambda}})\gamma_2\left(R^{\intercal}_{T_{n},\boldsymbol{\lambda}}\right)
            + 1 \, , \\
            \le&~ 2(2n^2-1)T_n(1 + \delta) \gamma_2(R_{T_{n+1},\boldsymbol{\lambda}}) + \\
               &~ 2\left(2(n+1)^2-1\right)T_{n+1}(1 + \delta)\gamma_2(R^{\intercal}_{T_{n},\boldsymbol{\lambda}})
            + 1 \, , \\
             \le&~ 2(2n^2-1)T_n(1 + \delta) T_{n+1}(1+\delta) + \\
                &~ 2\left(2(n+1)^2-1\right)T_{n+1}(1 + \delta) T_n(1+\delta) + 1 \, , \\
             \le&~ 2\left((2n+1)^2-1\right)  T_{n}(1+\delta)T_{n+1}(1+\delta) + 1 \, .
        \end{aligned}
    \end{equation}
    Next using $2T_n(1+\delta) T_{n+1}(1+\delta) = T_{2n+1}(1+\delta) + 1 + \delta$ and the fact that $T_{2n+1}(1+\delta) \ge 1+\delta$ we obtain
    \begin{equation}
        \begin{aligned}
            \gamma_2(\mathcal D_{T_{2n+1},\boldsymbol{\lambda}}) 
             &\le \left((2n+1)^2-1\right)  \left(T_{2n+1}(1+\delta) + 1 + \delta \right) + 1 \, , \\
             &\le \left(2(2n+1)^2-1\right)  T_{2n+1}(1+\delta) \, .
        \end{aligned}
    \end{equation}
\end{proof}

We can now state the main lemma that will be used in the proof of \cref{thm: expansion-like}.

\begin{lem}\label{lem: GammaNormCm}
For any $n$ define the polynomial $C_n(x)$ by 
\begin{equation}
    C_n(x) := 1 - \frac{T_n(f(x))}{T_n(f(0))}\, ,
\end{equation}
where 
\begin{equation}
    f(x) = \frac{1+\epsilon - 2x}{1 - \epsilon}\,,
\end{equation}
and $0 \le \epsilon \le 1/3$, then for any sequence $\boldsymbol{\lambda}=(\lambda_1, \dots, \lambda_d)$ with $0\leq \lambda_i\leq 1$ we have
\begin{align}
    \gamma_2(\mathcal D_{C_n,\boldsymbol{\lambda}}) &\le 6n^2 - 3
\end{align}
\end{lem}

\begin{proof}
    One can easily verify that
    \begin{equation}
        \begin{aligned}
            \mathcal D_{C_n,\boldsymbol{\lambda}} &= -\frac{1}{T_n(f(0))} \mathcal D_{T_n \circ f,\boldsymbol{\lambda}} \\
            &= -\frac{1}{T_n(f(0))} \mathcal D_{T_n,f(\boldsymbol{\lambda})} \circ \mathcal D_{f,\boldsymbol{\lambda}} \\
            &= -\frac{1}{T_n(f(0))} \mathcal D_{T_n,f(\boldsymbol{\lambda})} \circ (\frac{-2}{1 - \epsilon} \vb J) \\
        \end{aligned}
    \end{equation}
Where $\vb J$ is an all-one matrix.

Then we introduce $\delta := \frac{2\epsilon}{1-\epsilon}$ and we have  $f(0) = 1 + \delta$
\begin{equation}
\begin{aligned}
    \gamma_2(\mathcal D_{C_n,\boldsymbol{\lambda}}) &= -\frac{1}{T_n(1 + \delta)} \gamma_2(\mathcal D_{T_n,f(\boldsymbol{\lambda})}) (\frac{-2}{1 - \epsilon}) \\
    &\le \frac{1}{T_n(1 + \delta)} \gamma_2(\mathcal D_{T_n,f(\boldsymbol{\lambda})})(\frac{2}{1-\frac{1}{3}}) 
\end{aligned}
\end{equation}

Knowing $f(\lambda_i) \le 1+\delta$ and $0 \le \delta \le 1$, by applying \cref{lem: gamma_2 inequality} \begin{equation}
\begin{aligned}
    \gamma_2(\mathcal D_{C_n,\boldsymbol{\lambda}}) &\le \frac{3}{T_n(1 + \delta)} (2n^2 - 1)T_n(1 + \delta) \\
    \gamma_2(\mathcal D_{C_n,\boldsymbol{\lambda}}) &\le 6n^2 - 3
\end{aligned}
\end{equation}
\end{proof}

\section{A general isoperimetric inequality}
\label{app: proof-expansion-like}

In this appendix we prove \cref{thm: expansion-like}. This is a general isoperimetric inequality that will be used in the proof of \cref{thm: layers}.

\begin{proof}

Let $N_1$ be the set of qubits within distance $L$ from qubits $N_0=\{1, \dots, n_0\}$ in the graph, i.e.,
\begin{equation}N_1= N_0\cup \partial_L(N_0).\end{equation}
Observe that using~\eqref{eq:neighborhood-L-bound} we have $n_1=|N_1|\leq  \min\{2n_0 \Delta^{L}, n\}$. For simplicity of presentation we assume that $N_1=\{1,\dots, n_1\}$. 
Define
\begin{equation} 
\vb P_0 := \ketbra{s_1}\otimes \cdots \ketbra{s_{n_1}}\otimes\vb I^{\otimes (n-n_1)}, 
\end{equation}
and
\begin{equation} 
\vb \Gamma_0 := \frac{1}{n_1} \sum_{i =1}^{n_1} \ketbra{\bar s_i}_i, 
\end{equation}
where $\ket{\bar s_i}$ is the qubit state orthogonal to $\ket{s_i}$ and $\ketbra{\bar s_i}_i$ is the operator acting on the $i$-qubit by projecting on $\ket{\bar s_i}$. 

Let $T_m(x)$ for 
$m = \frac{1}{2} n_1^{\frac 12 - \theta}$
be the Chebyshev polynomial of degree $m$ given by \cref{eq:Chebyshev-poly}. Let
\begin{equation}
    f(x) = \frac{1+\frac{1}{n_1} - 2x}{1 - \frac{1}{n_1}},
\end{equation}
and define 
\begin{equation}
    C_m(x) := 1 - \frac{T_m(f(x))}{T_m(f(0))}.
\end{equation}
Next, define
\begin{equation} 
\vb K := \vb U C_m(\vb \Gamma_0) \vb U^\dagger = C_m\big(\vb U \vb \Gamma_0 \vb U^\dagger\big) \, ,
\end{equation}
where $\vb U=\vb U(T)$. We note that $C_m(0)=0$, and for $0\leq x\leq 1 $ we have 
\begin{equation}0\leq C_m(x)\leq 1+\frac{1}{T_m(f(0))}\leq 2.\end{equation}
Moreover, as shown in~\cite{Eldar2017} we have
\begin{equation}C_m(x)\geq 1-\frac{1}{1+2\frac{m^2}{n_1}}\geq \frac{m^2}{n_1}=\frac{1}{4} n_1^{-2\theta}, \qquad \frac{1}{n_1}\leq x\leq 1.\end{equation}
Note that the spectrum of $\vb \Gamma_0$, regardless of multiplicities, equals $\{0, \frac{1}{n_1}, \frac{2}{n_1}\dots, 1\}$. Thus, by the above properties of $C_m(x)$, the spectrum of $\vb K$ belongs to $\{0\}\cup [n_1^{-2\theta}/4,\, 2]$. Moreover, the projection on the eigenspace of $\vb K$ associated to the $0$-eigenvalue equals $\vb P=\vb U \vb P_0\vb U^{\dagger}$. Therefore, 
\begin{align}\label{eq:P-K-P}
\frac{1}{4} n_1^{-2\theta} (\vb I - \vb P) \leq \vb K\leq 2(\vb I- \vb P). 
\end{align}

The next step is to approximate $\vb K$ with a sum of local operators. Based on \cref{thm: lieb-robinson}, for $1\leq i\leq n_1$ let $\vb Q_i$ be an $L$-local operator such that 
\begin{equation}\big\|\vb Q_i -  \vb U \ketbra{\bar s_i}_i \vb U^\dagger \big\|\leq \epsilon(L),\end{equation}
where 
\begin{align}\label{eq:eps-L}
\epsilon(L) = \sqrt{\frac{2}{\pi}}\e^{-L\big(\log L-\log T - \log (4g(\Delta-1))\big)-\frac 12\log L}.
\end{align}
Indeed, $\vb Q_i$ acts on qubits within distance $L$ from $i$ in the graph. 
Then letting
\begin{equation}\tilde{\vb \Gamma}:= \frac{1}{n_1} \sum_{i=1}^{n_1} \vb Q_i,\end{equation}
we have  
\begin{equation}\big\| \tilde{\vb \Gamma} - \vb U\vb \Gamma_0 \vb U^\dagger\big\|\leq \epsilon(L).\end{equation}
We claim that $\tilde{\vb K}= C_m(\tilde{\vb \Gamma})$ is close to $\vb K= C_m\big( \vb U\vb \Gamma_0 \vb U^\dagger\big)$. To prove this we use \cref{lem: gamma-2-norm} and \cref{lem: GammaNormCm}. To apply the latter lemma we need to have a bound on the spectrum of $(1- r) \vb U \vb \Gamma_0 \vb U^{\dagger} +r\tilde{\vb \Gamma}$. Here, a crucial observation is that by construction of $\vb Q_i$ in the proof of \cref{thm: lieb-robinson} we have $0\preceq \vb Q_i\preceq \vb I$ which gives $0\preceq \tilde{\vb \Gamma}\preceq 1$.  On the other hand, by definition $0\preceq \vb \Gamma_0\preceq \vb I$ which yields $0\preceq (1- r) \vb U \vb \Gamma_0 \vb U^{\dagger} +r\tilde{\vb \Gamma}\preceq \vb I$ for $0\leq r\leq 1$. As a result, by \cref{lem: gamma-2-norm} and \cref{lem: GammaNormCm} we have
\begin{equation}\label{eq:tilde-K-approx}
    \big\|{\vb{\widetilde{K}}} - \vb K\big\| \leq (6m^2-3)  \epsilon(L)\leq 6m^2  \epsilon(L).
\end{equation}
Now recall that $\tilde{\vb \Gamma}$ is a sum of $L$-local operators, and $C_m(x)$ is a degree $m$ polynomial. As a result  $\tilde{\vb K}= C_m(\tilde{\vb \Gamma})$ is a sum of  $\ell'$-local operators where 
$$\ell'= mL.$$

Next, as in~\cite{Eldar2017} we partition $\{0,1\}^n$ into four sets:
\begin{align}
F_1 &:= \big(F\setminus \partial_{\ell'} (F))\times \{0,1\}^{n-n_0}, \\
F_2 &:= \big(F\cap \partial_{\ell'} (F))\times \{0,1\}^{n-n_0},  \\
F_3 &:= \big(F^c\cap \partial_{\ell'} (F))\times \{0,1\}^{n-n_0}, \\
F_4 &:= \big(F^c\setminus \partial_{\ell'} (F))\times \{0,1\}^{n-n_0}.
\end{align}
Here, $\partial_{\ell'}(F)$ is defined with respect to the Hamming distance on $\{0,1\}^{n_0}$. 
Then, expanding $\ket{\psi} = \ket{\psi_T}$ in the computational basis, we can write
\begin{equation}\ket{\psi} = \ket{\phi_1}+\ket{\phi_2}+\ket{\phi_3}+\ket{\phi_4},\end{equation}
where $\ket{\phi_j}$ is a linear combination of $\ket{\vb x}$'s with $\vb x\in F_j$. The point of this decomposition is that $p(\partial_{\ell'} (F)) = \|\phi_2\|^2+\|\phi_3\|^2$. On the other hand, 
since $\vb{\widetilde{K}}$ is $\ell'$-local we have 
\begin{align}\label{eq:tilde-K-phi-13}
\mel{\phi_1}{\vb{\widetilde{K}}}{\phi_3} = \mel{\phi_1}{\vb{\widetilde{K}}}{\phi_4} =
\mel{\phi_2}{\vb{\widetilde{K}}}{\phi_4} = 0.
\end{align}
Therefore, because $\vb \Gamma_0\ket{\psi_0}=0$, we have 
\begin{equation}
     \ev{\vb K}{\psi} = \bra{\psi} \vb UC_m(\vb \Gamma_0) \vb U^\dagger \ket{\psi} =  \bra{\psi_0} C_m(\vb \Gamma_0)  \ket{\psi_0}=0 .
\end{equation}
Thus by \cref{eq:tilde-K-approx} we have 
\begin{align}\label{eq:tilde-K-psi}
\big|\ev{\tilde {\vb K}}{\psi}\big|\leq 6m^2  \epsilon(L). 
\end{align}

Let $ \ket{\psi'} := - \ket{\phi_1} - \ket{\phi_2} + \ket{\phi_3} + \ket{\phi_4}$. Then using \cref{eq:tilde-K-phi-13} we have
\begin{align}
\bra{\psi'}\tilde{\vb K}\ket{\psi'} & = \sum_{(i, j)\in \{1, 2\}^2\cup \{3, 4\}^2} \bra{\phi_i}\tilde{\vb K}\ket{\phi_j} - \sum_{(i, j)\in \{1, 2\}\times \{3, 4\}}   \bra{\phi_i}\tilde{\vb K}\ket{\phi_j}\\
& = \|\phi_1\|^2 + \|\phi_2\|^2 + \|\phi_3\|^2 + \|\phi_4\|^2 + 2\Re \bra{\phi_1}\tilde{\vb K}\ket{\phi_2} + 2\Re \bra{\phi_3}\tilde{\vb K}\ket{\phi_4} - 2\Re \bra{\phi_2}\tilde{\vb K}\ket{\phi_3}.
\end{align}
Considering the same expansion for $\bra{\psi}\tilde{\vb K}\ket{\psi}$ and using \cref{eq:tilde-K-psi}, we obtain
\begin{equation} 
\label{eq:k_upperbound}
\begin{aligned}
\ev{\tilde{\vb K}}{\psi'} &\le 4 \big|\mel{\phi_{2}}{\tilde{\vb K}}{\phi_{3}}\big|  +  6m^2  \epsilon(L)\\
&\le 4\big\|\vb{\widetilde{K}}\big\|\cdot \norm{\ket{\phi_2}}\cdot \norm{\ket{\phi_3}}
+ 6m^2  \epsilon(L) \\
&\le 4 \left( \norm{\ket{\phi_2}}^2 + \norm{\ket{\phi_3}}^2 \right)
+ 6m^2 \epsilon(L) \\
&= 4 p(\partial_{\ell'} (F))
+ 6m^2  \epsilon(L),
\end{aligned}
\end{equation}
where in the third line we use $0\preceq \tilde{\vb \Gamma}\preceq \vb I$ which gives $\|\tilde{\vb K}\| = \|C_m(\tilde{\vb \Gamma})\|\leq 2$.

We now prove a lower bound on $\ev{\tilde{\vb K}}{\psi'}$ in terms of $p(F)$. To this end we will use \cref{eq:P-K-P}, so we first prove an upper bound on $\ev{\vb P}{\psi'}$. Define
\begin{equation}\vb R_0:= \sum_{\vb x\in F^c} \ketbra{\vb x}\otimes \vb I  -\sum_{\vb x\in  F} \ketbra{\vb x}\otimes \vb I.\end{equation}
Note that $\vb R_0\ket{\psi}= \ket{\psi'}$ and that $\vb R_0$ is an $n_0$-local operator acting on the first $n_0$ qubits. Then by \cref{thm: lieb-robinson} there is $\tilde {\vb R}$ acting on qubits in $N_1$ such that 
\begin{equation}\big\| \tilde{\vb R} - \vb R_0    \big\|\leq n_0\epsilon(L).\end{equation}
Then we have
\begin{align}
\vb P_0 \vb U^{\dagger}\ket{\psi'} & = \vb P_0 \vb U^\dagger \vb R_0\ket{\psi} \\
& = \vb P_0 \vb U^\dagger \vb R_0\vb U\ket{\psi_0}\\
& = \vb P_0 \tilde{\vb R}\ket{\psi_0} + \ket{\delta},
\end{align}
where $\ket{\delta}$ is a vector with $\|\ket{\delta}\| \leq n_0\epsilon(L)$. On the other hand, since $ \tilde{\vb R}$ acts only on the first $n_1$ qubits and $\ket{\psi_0} = \ket{s_0}\cdots \ket{s_n}$, the vector $\tilde{\vb R}\ket{\psi_0}$ is a linear combination of vectors of the form $\ket{\vb x}\otimes \ket{s_{n_1+1}}\cdots \ket{s_n}$ for $\vb x\in \{0,1\}^{n_1}$. Therefore, applying the projection $\vb P_0$ on $\tilde{\vb R}\ket{\psi_0}$ we get a vector parallel to $\ket{\psi_0}$, i.e.,
$\vb P_0\tilde{\vb R}\ket{\psi_0} = \alpha \ket{\psi_0}$ for some $\alpha\in \mathbb C$. This means that
\begin{align}
\vb P_0 \vb U^{\dagger}\ket{\psi'} = \alpha\ket{\psi_0} + \ket{\delta}.
\end{align}
To compute $\alpha$ we write
\begin{align}
\alpha &= \bra{\psi_0}\vb P_0\vb U^\dagger\ket{\psi'} - \langle \psi_0|\delta\rangle \\
&= \bra{\psi_0}\vb P_0\vb U^\dagger \vb R \vb U  \ket{\psi_0} - \langle \psi_0|\delta\rangle\\
&= \bra{\psi_0} \vb U^\dagger \vb R \vb U  \ket{\psi_0} - \langle \psi_0|\delta\rangle\\
&= \bra{\psi}  \vb R \ket{\psi} - \langle \psi_0|\delta\rangle\\
&= \bra{\psi} \psi'\rangle - \langle \psi_0|\delta\rangle\\
& = -\|\ket{\phi_1}\|^2 -\|\ket{\phi_2}\|^2 +\|\ket{\phi_3}\|^2 +\|\ket{\phi_4}\|^2 - \langle \psi_0|\delta\rangle\\
& = (1-2p(F)) - \langle \psi_0|\delta\rangle,
\end{align}
where in the last line we have used $p(F)=\|\ket{\phi_1}\|^2 +\|\ket{\phi_2}\|^2$. Putting these together we arrive at 
\begin{equation}\vb P_0 \vb U^{\dagger}\ket{\psi'}  = \big((1-2p(F)) - \langle \psi_0|\delta\rangle\big)\ket{\psi_0} + \ket{\delta}.\end{equation}
Therefore,
\begin{align}
\bra{\psi'} \vb P\ket{\psi'} & = \bra{\psi'}\vb U\vb P_0\vb U^\dagger\ket{\psi'}\\
& = \big\|\vb P_0 \vb U^{\dagger}\ket{\psi'} \big\|^2\\
&\leq \big|(1-2p(F)) - \langle \psi_0|\delta\rangle\big|^2 + \|\ket{\delta}\|^2 + 2\big|(1-2p(F)) - \langle \psi_0|\delta\rangle\big|\cdot \|\ket{\delta}\|\\
& \leq (1-2p(F))^2 + 8\|\ket{\delta}\|\\
& \leq (1-2p(F))^2  + 8 n_0\epsilon(L).
\end{align}
Next, using \cref{eq:P-K-P} we find that 
\begin{align}
\bra{\psi'}\vb K\ket{\psi'} &\geq \frac{1}{4}n_1^{-2\theta} \big(1 - \bra{\psi'}\vb P\ket{\psi'}\big)\\
&\geq  \frac{1}{4}n_1^{-2\theta} \Big(1 - (1-2p(F))^2  - 8 n_0\epsilon(L) \Big)\\
& \geq \frac{1}{2} n_1^{-2\theta} p(F)  - 2n_1^{-2\theta} n_0\epsilon(L),
\end{align}
where in the last line we used $p(F)\leq 1/2$ to conclude that $1-(1-2p(F))^2\ge 2p(F)$.
Then by \cref{eq:tilde-K-approx} we have
\begin{align}
\big|\bra{\psi'}\tilde{\vb K}\ket{\psi'}\big| \geq \frac{1}{2} n_1^{-2\theta} p(F)  - \big(2n_1^{-2\theta}n_0 + 6m^2\big) \epsilon(L).
\end{align}
Finally, comparing this inequality with \cref{eq:k_upperbound} yields
\begin{align}
4 p(\partial_{\ell'} (F))
+ 6m^2  \epsilon(L)\geq \frac{1}{2} n_1^{-2\theta} p(F)  - \big(2n_1^{-2\theta}n_0 + 6m^2\big) \epsilon(L),
\end{align}
or equivalently 
\begin{align}\label{eq:isoper-bound-ell-prime}
p(\partial_{\ell'} (F))
\geq \frac{1}{8} n_1^{-2\theta} p(F)  - \Big(\frac{1}{2}n_1^{-2\theta}n_0 + 3m^2\Big) \epsilon(L).
\end{align}

Now for the first part of the theorem
let $L=\frac{k_1}{\log \Delta}\log n_0 =\beta 4g\Delta T $ with
\begin{equation}
\beta = \frac{\kappa_1}{4g\Delta\log (\Delta)  T}\log n_0\geq \Delta^{\frac{1+\kappa_1+\kappa_2}{\kappa_1}},
\end{equation}
where the inequality follows from the assumption 
\begin{equation}T\leq \frac{\kappa_1}{4g\Delta\log (\Delta)\Delta^{\frac{1+2\kappa_1+\kappa_2}{\kappa_1}}}\log n_0.\end{equation}
Then we have $n_1\leq2 n_0\Delta^L=2n_0^{1+\kappa_1}$ and  $m=\frac{1}{2}n_1^{\frac 12-\theta} \leq \frac{1}{\sqrt 2}n_0^{\frac {1}{2}(1+\kappa_1)}$. We also have
\begin{align}
\epsilon(L) &=\sqrt{\frac{2}{\pi}}\e^{-L\big(\log L-\log T - \log (4g(\Delta-1))\big)-\frac 12 \log L}\\
&\leq \e^{-L\log \beta}\\
& =  n_0^{-\frac{\kappa_1}{\log (\Delta)}\log \beta}\\
&\leq n_0^{- (1+\kappa_1+\kappa_2)}.
\end{align}
Therefore,
\begin{align}
 \Big(\frac{1}{2}n_1^{-2\theta}n_0 + 3m^2\Big) \epsilon(L) & \leq  2 \,n_0^{1+\kappa_1}n_0^{- (1+\kappa_1+\kappa_2)}= 2n_0^{-\kappa_2}.
 \end{align}
On the other hand,
\begin{equation}
\ell' = mL\leq \frac{1}{2}(2n_0^{1+\kappa_1})^{\frac 12-\theta}\frac{\kappa_1}{\log (\Delta)}\log n_0\leq \ell.\end{equation}
We conclude that
\begin{align}
p(\partial_{\ell} (F))\geq p(\partial_{\ell'} (F))
\geq \frac{1}{8} n_1^{-2\theta} p(F)  - 2n_0^{-\kappa_2} \geq  \frac{1}{8} (2n_0^{1+\kappa_1})^{-2\theta} p(F)  -2 n_0^{-\kappa_2}.
\end{align}

For the second part of the theorem, note that when $n_0=n$, then $n_1=n$. Using this in~\eqref{eq:isoper-bound-ell-prime} for $\theta=0$, we obtain
\begin{equation}p(\partial_{\ell'})\geq \frac 18 p(F) - \frac 54 n\epsilon(L).\end{equation}
Let $L=\kappa_1(1+\kappa_2)\log n$ and note that $\ell'= mL = \frac 12 \sqrt n L=\ell$.
On the other hand, using the given bound  on $T$ we find that $\epsilon(L)\leq n^{-\kappa_1(1+\kappa_2)\log(1+\kappa_2)}$. Putting this in the above inequality yields the desired result.

\end{proof}

\section{Proof Theorem~\ref{thm: layers}}\label{app:layers}

For an integer $0\leq d< D/\ell$ define  $K_d\subseteq \{0,1\}^n$ by
    \begin{equation}
        K_d = \{\vb x \,:\, (d-1)\ell< \dist_H(\vb x, F_1)\leq d\ell \}\,,
    \end{equation}
where $\dist_H(\vb x, F_1)=\min_{\vb y \in F_1} \dist_H(\vb x, \vb y)$. 
Clearly, $K_0 = F_1$ and $K_d \cap F_2 = \emptyset$ for all  $d < D/\ell$.
Then, since $K_d$'s are disjoint, we have
    \begin{equation}
        \sum_{\stackrel{1\leq d< \frac D\ell-1}{d:\text{odd}}} p(K_d)+p(K_{d+1}) \leq 1 - \big(p(F_1)+p(F_2)\big)\leq 1- 2\mu \, .
    \end{equation}
Therefore, there exists $d_0$ such that 
    \begin{equation} 
    p(K_{d_0})+p(K_{d_0+1}) \leq \frac{1 - 2\mu}{ \frac{D}{2\ell} -2}\, .
    \label{Eq:upperbound_for_boundary}
    \end{equation}
Now define
    \begin{align}
        \widetilde F_1 &:= \bigcup_{j \leq d_0} K_d\, .
    \end{align}
Observe that by definition $F_1\subseteq \widetilde F_1$ and that $\partial_\ell (\widetilde F_1)\subseteq K_{d_0}\cup K_{d_0+1}$. Therefore, using \cref{thm: expansion-like} we obtain
\begin{align}    
\frac{1 - 2\mu}{ \frac{D}{2\ell} -2}&\geq  p(K_{d_0})+p(K_{d_0+1})\\
& \geq  p\big(\partial_\ell \widetilde F_1\big)\\
&\geq  \frac{1}{8}  p(\widetilde F_1)  - \frac 54 n^{-\big(\kappa_1(1+\kappa_2)\log(1+\kappa_2)-1\big)}\\
&\geq  \frac{\mu}{8}    - \frac 54 n^{-\big(\kappa_1(1+\kappa_2)\log(1+\kappa_2)-1\big)},
\end{align}
and
\begin{align}
    \frac{D}{2\ell}-2&\leq \frac{1-2\mu}{\mu/8} + \frac 54 n^{-\big(\kappa_1(1+\kappa_2)\log(1+\kappa_2)-1\big)}\Big(\frac{\frac{D}{2\ell}-2}{\mu/8}\Big)\\
    &< \frac{8}{\mu} -2 + \frac 54 n^{-\big(\kappa_1(1+\kappa_2)\log(1+\kappa_2)-1\big)}\Big(\frac{\frac{D}{2\ell}}{\mu/8}\Big).
\end{align}
Multiplying both sides by $2\ell$ yields the desired inequality.

\section{Proof of Theorem~\ref{thm:MaxCut-Ramanujan}}\label{app:MaxCut-Ramanujan}

We follow similar steps as in~\cite{Bravyi2019} to prove this theorem. Let $\vb x\in \{0,1\}^n$ be the measurement outcome on $\ket{\psi_T}$ in the computational basis. Let $\Cut(\vb x)$ be the size of the cut associated to $\vb x$:
\begin{equation}\Cut(\vb x) = -\bra{\vb x} \vb C\ket {\vb x}.\end{equation}
Then we have
\begin{equation}-\bra{\psi_T}\vb C\ket{\psi_T} = \mathbb E[\Cut(\vb x)].\end{equation}
Here a crucial observation is that $\Cut(\vb x) = \Cut(\bar{\vb x})$ where $\bar{\vb x}\in \{0,1\}^n$ is obtained from $\vb x$ by flipping each coordinate: $\bar{\vb x}_i=1-\vb x_i$. On the other hand, the distribution of $\bar {\vb x}$ is the same as the distribution of $\vb x$. This is because $\vb X^{\otimes n}\ket{\psi_T} = \ket{\psi_T}$  which itself can be proven using the fact that the starting state of QA satisfies $\vb X^{\otimes n}\ket{\psi_0}= \ket{\psi_0}$ and that $\vb X^{\otimes n}$ commutes with $\vb H(t)$ for any $t$. 

Let $\vb x^*$ be the configuration associated to the bipartition of  $G$. Then we have $\Cut^*=\Cut(\vb x^*) = \Cut(\bar {\vb x}^*) = |E| = \Delta n/2$. Let $d=\alpha n$ and define
\begin{equation} 
F_1 =  \{ \vb x \,:\, \dist_H(\vb x, \vb x^*)\leq d \},
\end{equation}
where as before, $\dist_H(\vb x, \vb x^*)$ denotes the hamming distance between $\vb x$ and $\vb x^*$. Similarly, define
\begin{equation} 
F_2 =  \{ \vb x\,:\, \dist_H(\vb x,  \bar{\vb x}^*)\leq d \} = \{\vb x\,:\, \bar{\vb x}\in F_1\}.
\end{equation}
Now using \cref{app:layers} and the fact that $D := \dist_H(F_1, F_2) \geq n - 2d $ we have
\begin{align}
p(F_1)=p(F_2)=\min\{p(F_1), p(F_2)\} 
\leq \frac{16 \ell}{D} + 10 n^{-\kappa_3},
\end{align}
where the first equality follows from the aforementioned symmetry and $\ell$ is given by \cref{eq:ell:n_0=n} for some $\kappa_2>0$ and 
\begin{equation}\kappa_3=\kappa_1(1+\kappa_2)\log(1+\kappa_2)-1.\end{equation}

Suppose that $\vb x$ is not in $F_1\cup F_2$. This means that either $d<\dist_H(\vb x, \vb x^*)\leq n/2$ or $d<\dist_H(\vb x, \bar{\vb x}^*)\leq n/2$. Then, e.g., in the first case, using the definition of the Cheeger constant we have
\begin{equation}
    \Cut(\vb x) \leq  |E| - h(G)d < \Cut(\vb x^*) .
\end{equation}
This means that, letting
\begin{equation} J = \{ \vb x \,:\, \Cut(\vb x) \geq |E| - h(G)d\},\end{equation}
we have $J \subseteq F_1 \cup F_2$ and
\begin{equation} 
p(J) \leq p(F_1) + p(F_2) \le 2\Big(\frac{16 \ell}{D} + 10n^{-\kappa_3}\Big).\end{equation} 

Therefore,
\begin{equation} p\Big(|E| - \Cut(\vb x) > h(G)d\Big) \geq 1 - \frac{32 \ell}{D} - 20n^{-\kappa_3}. \end{equation}
Next,
using Markov's inequality we have
\begin{equation} \mathbb{E}\big[|E| - \Cut(\vb x)\big] \geq h(G)d \left(
1 - \frac{32 \ell}{D} - 20n^{-\kappa_3}
\right).\end{equation}
Now, note that $D\geq n-2d=(1-2\alpha)n$ and that for \emph{sufficiently large} $n$
we have 
\begin{equation}1 - \frac{32 \ell}{D} - 20n^{-\kappa_3}\geq 1-\epsilon.\end{equation}

Then, for any such $n$ we have
\begin{equation} \mathbb{E}[\Cut(\vb x)] \leq |E| - \alpha (1-\epsilon)h(G) n.\end{equation}
Thus, using $|E| = \frac{n\Delta}{2}$ and $h(G) \ge \frac{\Delta - 2\sqrt{\Delta - 1}}{2}$
we obtain
\begin{equation}\begin{aligned} \frac{\mathbb{E}[\Cut(\vb x)]}{\Cut^*} &\leq  1 - \alpha(1-\epsilon)\Big( \frac{n\Delta}{2|E|} - \frac{n\sqrt{\Delta - 1}}{|E|}\Big) \\ &= 1-\alpha(1-\epsilon) +2\alpha(1-\epsilon) \frac{\sqrt{\Delta - 1}}{\Delta}.\end{aligned}\end{equation}

\section{Proof of Theorem~\ref{thm:random-bipartite}}\label{app:random-bipartite}

Let $\mathbb{G}_\Delta(n)$ be a random $\Delta$-regular graph on $n$ vertices and $\mathbb{G}_\Delta^B(n)$ be a bipartite random $\Delta$-regular graph on $n$ vertices. 
Let $L$ be a parameter 
that grows at most logarithmically with $n$. 
Then, it is known~\cite{McKay2004} that there exists a constant $\alpha_\Delta$ independent of $n$ and $L$ such that 
\begin{align}
    \mathbb{E}_{\mathbb{G}_{\Delta} (n)}\left[\# \text{ of cycles of length at most } 2L+1\right] &\le \alpha_\Delta .\Delta^{2L+1}\, ,\\
    \mathbb{E}_{\mathbb{G}^B_{\Delta} (n)}\left[\# \text{ of cycles of length at most } 2L+1\right] &\le \alpha_\Delta .\Delta^{2L+1}\, .
\end{align}
Given a graph $G=(V, E)$ let $R_L(G)$ be the set of edges in $G$ that have a cycle in their $L$-neighborhood:
\begin{equation*}
    R_L(G) = \left\{e\in E \, |\, \exists \text{ a cycle in the } L\text{-neighborhood of } e \right\}.
\end{equation*}
It is easy to verify that if $e$ belong to $R_L(G)$, then there is cycle of length at most $2L+1$ in the $L$-neighborhood of $e$. On the other hand, for any such cycle, 
the number of edges in its $L$-neighborhood is at most:
\begin{equation}
    (2L+1)\Big[1+\Delta+\Delta(\Delta-1)+\cdots \Delta (\Delta-1)^{L-1}\Big]\le 2(2L+1)\Delta^{L}\, .
\end{equation}
Therefore, by the above bound on the expectation of the number of cycles of length at most $2L+1$ we have
\begin{align}
    \mathbb{E}_{\mathbb{G}_\Delta (n)}\left[\abs{R_L(\mathbb{G}_\Delta (n))}\right] &\le 2\alpha_\Delta (2L+1)\Delta^{3L+1}\leq \alpha_\Delta n^{\kappa}\, , \\
    \mathbb{E}_{\mathbb{G}_\Delta^B (n)}\left[\abs{R_L(\mathbb{G}_\Delta^B (n))}\right] &\le 2\alpha_\Delta (2L+1)\Delta^{3L+1}\leq \alpha_\Delta n^{\kappa} \, .
\end{align}

Let 
\begin{equation}L=\frac{\kappa}{4\log\Delta}\log n\, .\end{equation}
For an edge $e=\{i,j\}$ in $G$ define
\begin{equation}
     \mathcal{E}(e, L) =  \ev**{ \vb V_e(T)^\dagger \vb C_e \vb V_e(T) }{\psi_0}\, ,
\end{equation}
where 
\begin{equation}\vb C_e=\frac 12 (\vb I-\vb Z_i\vb Z_j)\, ,\end{equation}
and
$\vb V_e(T)$ is the unitary evolution associated to the Hamiltonian that includes the terms of $\vb H(t)$ that correspond to edges/vertices at distance at most $L$ from $e$. Note that by \cref{thm: lieb-robinson} we have
    \begin{align}
    \norm{\vb{\Upsilon}} \le \epsilon(L)= \sqrt{\frac{2}{\pi}} \e^{-L\big(\log L - \log T - \log (4\Delta)\big) - \frac12 \log L}\, ,
\end{align}
where
\begin{align}
    \vb{\Upsilon}_e =\vb U(T)^\dagger \mathbf C_e \vb U(T) - \vb V_e(T)^\dagger \mathbf C_e \vb V_e(T)\, .
    \end{align}
On the other hand,  observe that the $L$-neighborhood of any edge $e$ not in $R_L(G)$ is a tree, and that this is a fixed tree. Then, since $\vb V_e(T)$ depends only on the structure of the graph around $e$, and $\ket{\psi_0}$ is symmetric, for any $e\notin R_L(G)$, the quantity 
\begin{equation}\mathcal E(e, L) = \mathcal E_{\text{tree}}(L)\, ,\end{equation}
is independent of $e$.
Thus, using the fact that $\norm{\mathbf C_e}\le 1$, we have
\begin{align}
    -\ev**{ \vb{C} }{\psi_T}&=-\ev**{\vb U(T)^\dagger \mathbf C \vb U(T)}{\psi_0}\\ &= \sum_{e\notin R_L(G)}\ev**{\vb U(T)^\dagger \mathbf C_e \vb U(T)}{\psi_0} \nonumber\\
    &\quad +\sum_{e\in R_L(G)}\ev**{\vb U(T)^\dagger \mathbf C_e \vb U(T)}{\psi_0} \\
    &\le \abs{E}\cdot \big(\mathcal{E}_{\text{tree}}(L)+ \epsilon(L)\big) + \abs{R_L(G)} \\
    &\le \frac{n\Delta}{2}\big(\mathcal E_{\text{tree}}(L)+\epsilon(L)\big)+\abs{R_L(G)}\, . 
    \end{align}
Therefore,   
    \begin{align}
     \mathbb E_{\mathbb G^B_\Delta (n)}&\big[-\ev**{ \mathbf  C }{\psi_T}\big] \nonumber\\
     &\le \frac{n\Delta}{2}\big(\mathcal E_{\text{tree}}(L)+\epsilon(L)\big)+2\alpha_\Delta (2L+1) \Delta^{3L+1}\, . \label{eq:QA-exp-maxcut-bipartite}
\end{align}
On the other hand, we also have
\begin{align}
    -\ev**{ \mathbf C }{\psi_T}  &= \sum_{e\notin R_L(G)}\ev**{\vb U(T)^\dagger \mathbf C_e \vb U(T)}{\psi_0} \nonumber\\
    &\quad +\sum_{e\in R_L(G)}\ev**{\vb U(T)^\dagger \mathbf C_e \vb U(T)}{\psi_0} \\
    &\ge \big(\abs{E}-\abs{R_L(G)}\big)\big(\mathcal{E}_{\text{tree}}(L)-\epsilon(L)\big)\, ,
    \end{align}
and    
    \begin{align}
\mathbb E_{\mathbb G_\Delta (n)}&\left[-\ev**{\mathbf C }{\psi_T}\right] \nonumber\\
        &\ge \left(\frac{n\Delta}{2}-2\alpha_\Delta (2L+1) \Delta^{3L+1}\right)\big(\mathcal E_{\text{tree}}(L)-\epsilon(L)\big)\, .\label{eq:QA-exp-maxcut-general}
\end{align}

It is known that for any $\Delta$, there exists a constant $\rho_\Delta$ such that
\begin{equation}
    \mathbb{E}_{\mathbb{G}_\Delta (n)}\left[\text{size of }\textsc{MaxCut}\right] = \rho_\Delta n + o(n)\, .
\end{equation}
Therefore, since $-\ev**{\mathbf C }{\psi_T}$ itself is an expectation of sizes of certain cuts, we have
\begin{equation}
    \mathbb E_{\mathbb G_\Delta (n)}\left[-\ev**{\mathbf C }{\psi_T}\right] \le \rho_\Delta n+o(n)\, .
\end{equation}
Comparing this with \cref{eq:QA-exp-maxcut-general} yields
\begin{equation}
    \Delta\left(\frac{n}{2}-2\alpha_\Delta (2L+1) \Delta^{3L}\right)\big(\mathcal E_{\text{tree}}(L)-\epsilon(L) \big)\le \rho_\Delta n + o(n)\, .
\end{equation}
Next, using the definition of $L$ and the bound on $T$, we find that for any $\nu>0$ and sufficiently large $n$
 we have
 \begin{equation}
    \frac \Delta 2 \mathcal E_{\text{tree}} \le \rho_\Delta+\nu \, ,
\end{equation}
with 
\begin{equation}
    \mathcal E_{\text{tree}} = \limsup_{n\to \infty}\mathcal E_{\text{tree}}(L)\, ,
\end{equation}
where we note that $L$ can grow with $n$.
Using this in~\eqref{eq:QA-exp-maxcut-bipartite} we find that for sufficiently large $n$
\begin{equation}
    \mathbb E_{\mathbb G^B_\Delta (n)}\big[-\ev**{ \mathbf  C }{\psi_T}\big] \le (\rho_\Delta+2\nu)n\, .
\end{equation}

It is known that $\rho_3 \le 1.4026<\frac 3 2$,  and that for large $\Delta$~\cite{Kardos2012, Coppersmith2004,Dembo2017}
\begin{equation}
    \rho_\Delta \le \frac \Delta 4 + O(\sqrt{\Delta})\, .
\end{equation}
Thus, for sufficiently large $n$, short time QA on a random bipartite graphs gives a cut of size of at most \begin{equation}\left(\frac{\Delta}{4}+O(\sqrt{\Delta})+2\nu \right)n\, ,\end{equation} which is far from the optimal value of $\frac {\Delta}{2}n$.

\end{document}